\documentclass[a4paper, 12 pt, twoside, reqno]{amsart}

\usepackage[utf8]{inputenc}
\usepackage{slashed}
\usepackage[english]{babel}
\usepackage[all]{xy}
\usepackage{enumerate}
\usepackage{setspace}
\usepackage{hyperref}

\usepackage[euler-digits,euler-hat-accent]{eulervm}

\usepackage{amsthm}

\usepackage{times}
\usepackage{amsmath}
\usepackage{amsfonts}
\usepackage{amssymb}
\usepackage{amsthm}
\usepackage{graphicx}
\usepackage{array}
\usepackage{color}
\usepackage{mathrsfs}

\usepackage{eucal}
\usepackage{esint} 
\usepackage{tikz}
\usepackage{upgreek}

\usepackage{mathtools}

\allowdisplaybreaks

\theoremstyle{plain}
\newtheorem{theorem}{Theorem}[section]
\newtheorem{theorema}{Theorem}
\newtheorem{corollary}[theorem]{Corollary}
\newtheorem{proposition}[theorem]{Proposition}
\newtheorem{lemma}[theorem]{Lemma}

\theoremstyle{definition}
\newtheorem{definition}[theorem]{Definition}

\theoremstyle{remark}
\newtheorem{remark}[theorem]{Remark}

\numberwithin{equation}{section}
\numberwithin{figure}{section}
\numberwithin{table}{section}

\newcommand{\Ss}{\mathcal{S}}
\newcommand{\Tt}{\mathscr{T}}

\newcommand{\R}{\mathbb{R}}
\newcommand{\N}{\mathbb{N}}
\newcommand{\C}{\mathbb{C}} 
\newcommand{\Q}{\mathbb{Q}}
\newcommand{\Z}{\mathbb{Z}}
\newcommand{\T}{\mathbb{T}}

\newcommand{\G}{\mathcal{G}}
                 
\newcommand{ \ii}{\,\mathrm{i}\,}
\newcommand{\LL}{\mathcal{L}}

                \newcommand{\ie}{\textsl{i.\,e.\,}}

\DeclareMathOperator{\Tr}{Tr}

\newcommand{\X}{\mathcal{X}}

\newcommand{\A}{\mathfrak{A}}

\hypersetup{
pdftoolbar=true,        
pdfmenubar=true,        
pdffitwindow=true,     
pdfstartview=true,    
pdftitle={Invariant measures},    
pdfauthor={Danilo Polo Ojito},     
breaklinks=true, 
colorlinks=true,       
linkcolor=purple,         
citecolor=teal, 
urlcolor=blue, 
bookmarksopen=true, 
filecolor=magenta,      
}

\begin{document}

\title[Invariant measures on the transversal hull]{Invariant measures on the transversal hull of cone semigroups and some applications}

\author[D. Polo]{Danilo Polo Ojito}
\address{Department of Physics and Department of Mathematical Sciences, Yeshiva University 
	\\New York, NY 10016, USA \\
	\href{mailto:danilo.poloojito@yu.edu}{danilo.poloojito@yu.edu}}

\author[E. Prodan]{Emil Prodan}

\address{Department of Physics and
 Department of Mathematical Sciences 
\\Yeshiva University, New York, NY 10016, USA \href{mailto:prodan@yu.edu}{prodan@yu.edu}
}

\author[T. Stoiber]{Tom Stoiber}

\address{Department of Physics and Department of Mathematical Sciences, Yeshiva University 
	\\New York, NY 10016, USA \\
	\href{mailto:tom.stoiber@yu.edu}{tom.stoiber@yu.edu}}

\vspace{2mm}

\date{\today}

\maketitle

\begin{abstract}
 Let $\LL_{\bf v}\subset \Z^D$ be a suitable cone semigroup and $\A_{\bf v}$ its reduced semigroup $C^*$-algebra.  In this paper, we compute the $\LL_{\bf v}$-invariant measures in the transversal hull of the semigroup $\LL_{\bf v}$ that exhibit regularity in the boundaries of $\LL_{\bf v}.$  These measures enable the construction of a trace per-unit hypersurface for observables in $\A_{\bf v}$ supported near the boundaries of $\LL_{\bf v}$, leading to the construction of appropriate Chern cocycles in the "boundary" ideals of $\A_{\bf v}$.  Our approach applies to both finitely and non-finitely generated cone semigroups. Applications for the bulk-defect correspondence of lattice models of topological insulators are also provided


\end{abstract}
\section{Introduction}
Let ${\bf v}:=\{ v_1,\dots, v_d\}$ be a set of normalized linearly independent vectors in $\R^D$ with $D\geq d$, and consider $\A_{\bf v}=C^*_r(\LL_{\bf v})$ as the reduced semigroup $C^*$-algebra associated with the \emph{cone} subsemigroup $\LL_{\bf v}$ of $\Z^D$, with the latter defined as 
\begin{equation}\label{eq: cone semigroup}
   \LL_{\bf v}\;:=\;\bigcap_{i=1}^d\LL_{v_i} 
\end{equation}
where $\LL_{v_i}:=\{ n\in \Z^D\;|\; v_i\cdot n \geq 0\}.$ Note that $0\in \mathcal{L}_{\bf v}$, hence the semigroup has a unit. This semigroup is not, in general, finitely generated, since the vector components may be linearly independent over $\Q.$ Consequently, describing this $C^*$-algebra and computing its $K$-theory has posed significant challenges in recent years. Nevertheless, the case $D=2$ with $d=1,2$ (corresponding to Toeplitz and quarter-plane algebras) is now well understood, with a rich body of literature detailing its structure, classification, and
$K$-theory \cite{Dou,JI, JK, RX, Jia, Par1, Par2}. In solid state physics, the $C^\ast$-algebra 
$\A_{\bf v}$ relates to the dynamics of electrons in a crystal that has been etched in a multifaceted fashion. Specifically, all generators of such dynamics derive from representations of self-adjoint elements from $\A_{\bf v}$.

\medskip

It is known that $\A_{\bf v}$ agrees with the reduced $C^*$-algebra of a partial transformation groupoid \cite{CEL, Li, LAC, RS} (see also \cite[Ch.~5]{CELYBook}), \ie there is a partial transformation groupoid $\Xi_{\bf v}\!\Join\!\Z^D$ with $\Xi_{\bf v}$ a totally disconnected compact Hausdorff space endowed with a partial $\Z^D$-action such that
$$\A_{\bf v}\;\simeq \; C^*_r(\Xi_{\bf v}\!\Join\!\Z^D).$$
A convenient characterization of the space $\Xi_{\bf v}$ was exhibited in \cite{RS} as a Wiener-Hopf compactification, namely the so-called \emph{transversal hull} $\Xi_{\bf v}= \{\LL_{\bf v}-n: \; n\in \LL_{\bf v}\}^{\mathrm{cl}}$ where the closure is in the Fell topology of $\mathscr{C}(\Z^D)$, the set of closed subsets of $\Z^D$. Our results are based on a computation of $\Xi_{\bf v}$ when ${\bf v}$ satisfies the discrete or completely irrational (RCI) property; see Definition \ref{def: rational dependence}. Only if ${\bf v}$ is rational then $\LL_{\bf v}$ is finitely generated and the space $\Xi_{\bf v}$ consists of a countable set of points. Otherwise, this space will generally be uncountable. In particular, we shall prove that it admits a filtration
\begin{equation}\label{eq: filtration}
    \{ \Z^D\}\;=\;\Xi_{0}\subset \Xi_1\subset \cdots\subset \Xi_{d-1}\subset \Xi_{d}=\Xi_{\bf v}
\end{equation}
by closed subsets invariant under the natural semigroup action of $\LL_{\bf v}$ on $\Xi_{\bf v}$. The significance of this filtration lies in the fact that each  $\Xi_r\setminus \Xi_{r-1}$ for $r>0$ encodes the information of all boundaries of $\LL_{\bf v}$ with codimension $r$. More precisely,  the support of the induced multiplication operator in $\ell^2(\LL_{\bf v})$ by any $f\in C_c(\Xi_{r}\setminus \Xi_{r-1})$ is concentrated close to the boundaries of codimension $r$. Those results can be seen as the discrete analogue of similar computations for $C^*$-algebras of Wiener-Hopf operators on simplicial cones \cite{Alldridge,MuhlyRenault}. There, the transversal hull is itself homeomorphic to a cone which decomposes as a CW-complex, in contrast the topological spaces here are all completely disconnected.

\medskip

The filtration \eqref{eq: filtration} of $\Xi_{\bf v}$ induces a cofiltration of $\A_{\bf v}$
\begin{equation}
    \A_{\bf v}=\A_{d}\stackrel{\psi_d}\to \A_{d-1}\stackrel{\psi_{d-1}}\to \cdots \to \A_1\stackrel{\psi_1}\to \A_0\simeq C(\T^D)
\end{equation}
where $\psi_r$ are surjective $*$-homomorphisms and $\A_{r}:=C^*_r(\LL_{\bf v}|_{\Xi_r})$.  The $r$-codimensional boundary algebra is defined as $ \mathfrak{I}_r={\rm Ker}(\psi_r)\simeq C^*_r(\LL_{\bf v}|_{\Xi_r\setminus\Xi_{r-1}})$. By construction, we have for every $r>0$ an exact sequence
\begin{equation}\label{eq: sequence 1}
0\to \mathfrak{I}_r\to \A_r\to \A_{r-1}\to 0.
\end{equation}
In the analysis of topological insulators, one considers Hamiltonians on cone-like regions like $\LL_{\bf v}$, \ie, self-adjoint operators on $\ell^2(\LL_{\bf v})$ which have topological obstructions to the opening of spectral gaps: For $h\in \A_d$ one can consider the smallest $r$ such that the image $h_{r-1}$ of $h$ in $\A_{r-1}$ has a spectral gap and then associate to it a class $[h_{r-1}]_i\in K_i(\A_{r-1})$. There is then a natural connecting map $\partial_r\colon K_i(\A_{r-1})\to K_{1-i}(\mathfrak{I}_r)$ mapping invariants of Hamiltonians that are spectrally gapped on $\Xi_{r-1}$ to obstructions to spectral gap-opening for Hamiltonians on $\Xi_r$ (see \cite{DET1} for more details). 

Instead of abstract K-group elements, one generally prefers to indicate these topological invariants in terms of numerical invariants obtained by pairing the K-groups with cyclic cocycles. For rational ${\bf v}$, the $K_i(\mathfrak{I}_r)$ group elements can be indicated uniquely using a finite number of explicit cocycles, the so-called the \emph{Chern cocycles}. One of the main difficulties in generalizing those to rationally independent $\bf v$ lies in the construction of suitable densely defined lower semi-continuous traces on the ideals $\mathfrak{I}_r$. We use the strategy employed in \cite{NGP} to construct a trace on $\mathfrak{I}_r$ by finding an invariant Radon measure supported on the corresponding subset $\Xi_{r}\setminus \Xi_{r-1}$ of the unit space of the groupoid. Namely, we shall consider the vector space $\mathfrak{M}(\Xi_{\bf v})$ of all $\LL_{\bf v}$-invariant Borel measures on $\Xi_{\bf v}$ whose restriction to some $\Xi_r\setminus \Xi_{r-1}$ is a Radon measure. The explicit computation of this vector space, under some assumptions on ${\bf v}$, is the first main result of this work:
 
 \begin{theorema}\label{teo measures}
 Let ${\bf v}:=\{ v_1,\dots, v_d\}$ be a set of normalized linearly independent vectors in $\R^D$ with the $RCI$ property (see Definition \ref{def: rational dependence}). Then the vector space $\mathfrak{M}(\Xi_{\bf v})$  has dimension $2^d$ and, moreover,  there is a unique ergodic probability measure on $\Xi_{\bf v}$ (relative to the partial $\Z^D$-action).
\end{theorema}
\noindent
An explicit base $\{\mu_I\}_{I\in \mathtt{P}(\{1,\dots,d\})}$ for this space is provided in Remark \ref{rem: basis of measure}, where $ \mathtt{P}(\{1,\dots,d\})$ is the power set of $\{1,\dots,d\}$.  Each basis element $\mu_I$  for $I\neq \emptyset$ has support in $\Xi_r\setminus\Xi_{r-1}$ with $|I|=r$, while  $\mu_\emptyset$ is the unique ergodic probability measure and corresponds with the Dirac measure concentrated on the unique invariant point $\Xi_0=\{\Z^D\}$ of $\Xi_{\bf v}$. 

\medskip
Our second main result relies on this construction for the definition of Chern cocycles and on the proof of the non-triviality of the induced numerical invariant:

\begin{theorema}\label{Theorem 2} Under the assumptions of Theorem \ref{teo measures}, for each $|I|=r$ the linear functional
\[
\Tt_I(f)\;=\;\int_{\Xi_r\setminus \Xi_{r-1}}E(f)(x){\rm d}\mu_I(x),\qquad f\in \mathfrak{I}_r
\]
supplies a densely defined, faithful, and lower semi-continuous trace on $\mathfrak{I}_r$. Here $E\colon \mathfrak{I}_r\to C_0(\Xi_r\setminus \Xi_{r-1})$ is the standard conditional expectation. For a suitable $(m+1)$-tuple of elements $f_0,f_1,\dots,f_m\in \mathfrak{I}_r$ and a set of vectors ${\bf w}=\{ w_1,w_2,\dots,w_m\}$ in $\R^D$, the $(m+1)$-linear functional

\[
\begin{split}
     {\rm Ch}_{I,{\bf w}}(f_0,f_1,\dots, f_m)\;:=\;\sum_{\rho\in S_m}(-1)^\rho\Tt_I\big(f_0\nabla_{w_{\rho(1)}}f_1\cdots \nabla_{w_{\rho(m)}}f_m\big)
\end{split}
\]
defines a $m$-cocycle on $\mathfrak{I}_r$. Here the directional derivatives $\nabla$ are defined according to \eqref{eq: deriva}. Furthermore, the canonical pairing $\langle [u]_i,[{\rm Ch}_{I,{\bf w}}]\rangle$ with the $K$-groups of $\mathfrak{I}_r$ is a non-trivial numerical invariant. 
\end{theorema}

It is important to point out that the non-triviality of Chern cocyle given in Theorem \ref{Theorem 2}, provides partial and, in some cases full, information of the connecting map $\partial_r \colon  K_i(\A_{r-1})\to K_{1-i}(\mathfrak{I}_r)$ associated with the sequence \eqref{eq: sequence 1}. This is invaluable for the topological quantization of edge currents \cite{Bel1, Bel, NGP, Dani, PRO,Tom} and also for higher-order topological phases \cite{BenalcazarScience2017,DET1, DET2,SchindlerSciAdv2018}. We shall present a discussion of it at the end of Section \ref{Sec: chern} for the bulk-edge correspondence with irrational interfaces.

In Section~\ref{sec: trace}, we finally show that the traces we construct can be interpreted in the representation on $\ell^2(\LL_{\bf v})$ as the (Hilbert space) trace averaged with respect to the directions orthogonal to the boundaries. This allows one to relate the traces to physically relevant quantities like boundary currents and densities, which is needed to interpret the Chern cocycles as transport coefficients. We, however, leave such applications to the future.

\medskip

\noindent
{\bf Acknowledgements:} The authors would like to cordially thank G. De Nittis and J. Gomez for several stimulating discussions. This work was supported by the U.S. National Science Foundation through the grant CMMI-2131760, and by U.S. Army Research Office through contract W911NF-23-1-0127, and the German Research Foundation (DFG) Project-ID 521291358.

\section{Semigroup $C^*$-algebras}
In this section, we present standard definitions and results concerning semigroup $C^*$-algebras. Our exposition follows primarily \cite{Li,LAC,CELYBook}. We, however, emphasize a particular point of view in order to make the connection with the Bellissard-Kellendonk formalism \cite{Bel1, Kel1} explicit (see Remark~\ref{Re:BK}).

\smallskip

Let $(\mathcal{L},+)$ be an additive subsemigroup of  $\Z^D$. Its left-regular representation is carried by $\ell^2(\LL)$ and consists of the family of partial isometries $\LL\ni l\mapsto V_l$ which act via
\begin{equation}
    (V_l\psi)(x)\;:=\;\psi(x+l)\,,\qquad \psi\in \ell^2(\LL).
\end{equation}
The reduced semigroup $C^*$-algebra $C_r^*(\LL)$ of  $\LL$ is the $C^*$-algebra inside of $\mathscr{B}(\ell^2(\LL))$ generated by those partial isometries, \ie $C_r^*(\LL)\;:=\;C^*\big\{ V_l\,|\,l\in \LL\big\}$.

\smallskip 

A suitable ambient space for $\LL$ and all its possible configurations can be found inside the space  $\mathscr{C}(\Z^D)$ of closed subsets of $\Z^D$ endowed with the \emph{Fell topology}. 
The latter is a compact metric space \cite{Fell} where the metric is defined as follows: given $\LL\in \mathscr{C}(\Z^D)$ set
\begin{equation*}
    \LL(r)\;:=\;\LL\cap B(0,r).
\end{equation*}
Here, $B(0,r)\subset \R^D$ is the open ball centered in  $0$ with radius $r>0$.  Since $\Z^D$ is discrete, the Fell topology coincides with the Vietoris topology and is generated by the metric
\begin{equation}\label{eq: fell metric}
      {\rm D}(\LL,\LL')\;:=\;\inf\big\{ (r+1)^{-1}\;|\; \LL(r)=\LL'(r)\big\}.
\end{equation}
In particular, a sequence of sets in $(\Ss_n)_{n\in \N}$ in $\mathscr{C}(\Z^D)$ converges if and only if each element of $\Z^D$ is eventually contained either in each or none of the sets $\Ss_n$.

\begin{lemma}
There is a homeomorphism $\mathscr{C}(\Z^D)\simeq \{0,1\}^{\Z^D}$ with the product topology, i.e. it is in particular a totally disconnected space.
\end{lemma}
\begin{proof}
Clearly, $\Psi: M\in \mathscr{C}(\Z^D)\mapsto \chi_M$ is a bijection. Let us recall that a basis for the product topology on $\{0,1\}^{\Z^D}$ is given by cylinder sets of the form $$Z_{n,a}=\big\{f\in \{0,1\}^{\Z^D}: f(n)=a\big\},\qquad n\in\Z^D,\;a\in\{0,1\}$$
The pre-images are $\Psi^{-1}(Z_{n,1})=\{M \in \mathscr{C}(\Z^D): n \in M\}$  and $\Psi^{-1}(Z_{n,0})=\{M \in \mathscr{C}(\Z^D): n \notin M\}$.
The Fell topology is the hit-or-miss topology generated by the basic open sets $\{M \in \mathscr{C}(\Z^D): M\cap U\neq\emptyset, \; M\cap K= \emptyset\}$, where $U\subset \Z^D$ runs over all open (i.e. arbitrary since $\Z^D$ is discrete) and $K$ over all compact (i.e. finite) sets. It is easy to see that $\Psi^{-1}(Z_{n,a})$ is a basic open set and thus $\Psi$ is a continuous bijection. Since $\mathscr{C}(\Z^D)$ and $\{0,1\}^{\Z^D}$ both are compact Hausdorff spaces this means $\Psi$ is a homeomorphism.
\end{proof}

\begin{definition}
The \emph{transversal hull} of $\LL\in \mathscr{C}(\Z^D)$ is the compact Hausdorff space given by
\begin{equation}\label{eq: hull}
    \Xi_\LL\;:=\;\overline{\mathscr{O}(\LL)}\;\subset \; \mathscr{C}(\Z^D)
\end{equation}
    where the \emph{orbit space} is  $\mathscr{O}(\LL):=\{ \LL-n\,|\, n\in\LL\}$ and the closure is taken with respect to the Fell topology. 
    Similarly, we can also consider the $\Z^D$-hull of $\LL$ as
    \begin{equation}
        \tilde{\Xi}_\LL\;:=\;\overline{\mathscr{O}_{\Z^D}(\LL)}\setminus \emptyset\;=\;\overline{\{ \LL-n\,|\, n\in\Z^D\}}\setminus\emptyset
    \end{equation}
    This is a locally compact space for which the relation $\Xi_\LL\subset \tilde{\Xi}_\LL$ is fulfilled. 
\end{definition}

\medskip


There is a natural partial action  $\alpha$ of $\Z^D$, induced by $\LL$, on $\Xi_\LL$  provided by  the collection of open sets $\{ U_n\}_{n\in \Z^D}$ defined as
$$U_n\;:=\;\big\{\mathcal{S}\in \Xi_\LL\;|\; n\in \mathcal{S}\big\}\;\subset \;\Xi_\LL $$
with homeomorphism $\alpha_n\colon U_n\to U_{-n}$ given by $\alpha_n(\mathcal{S})=\mathcal{S}-n$ (see \cite[Definition 2.1]{Exel}). This homeomorphism is well-defined since $0\in \Ss$ for every $\Ss\in \Xi_\LL$. As a result, one gets the partial transformation topological groupoid
\begin{align}\label{eq: partial transformation groupoid}
   \G_\LL\;=\; \Xi_\LL\!\Join\! \Z^D\;:=\;\big\{ (\mathcal{S}, x)\in \Xi_\LL\times \Z^D\;|\; \mathcal{S}\in U_x\big\}
\end{align}
with topology inherited from $\Xi_\LL\times \Z^D$. It follows that the unit space $\G_\LL$  agrees with $\Xi_\LL\simeq \Xi_\LL\times \{0\}$ and  moreover comes equipped with 
\begin{enumerate}[(i)]
    \item inversion map $(\mathcal{S},x)^{-1}\;=\;(\mathcal{S}-x,-x)$
    \item source $\mathfrak{s}\colon \G_\LL\to \Xi_\LL$ and range $\mathfrak{r}\colon \G_\LL\to \Xi_\LL$ given by
    $$\mathfrak{s}\big((\mathcal{S},x)\big)\;=\;\mathcal{S},\qquad \mathfrak{r}\big((\mathcal{S},x)\big)\;=\;\mathcal{S}-x$$
    \item A set of composable elements $\G_\LL^{(2)}$ with multiplication
    $$ (\mathcal{S}-x,y)\cdot (\mathcal{S},x)\;=\;(\mathcal{S},x+y)$$
\end{enumerate}

\begin{remark}\label{Re:BK}
    As presented here, $\G_\LL$, which is the universal groupoid of $\LL$, coincides with the Bellissard-Kellendonk groupoid when $\LL$ is regarded as a uniformly separated pattern in the space $ \mathscr{C}(\R^D)$ of closed subsets of $\R^D$ endowed with the Fell topology \cite{Bel1, Kel1, DET1} (see also \cite[Sec.~4.2]{PatersonBook}). This is important because it is the Bellissard-Kellendonk groupoid that bridges mathematics and physics (see \cite{MP2024} for details). $\hfill \blacktriangleleft $
\end{remark}      

Some of the key properties of this groupoid are summarized below:

    \begin{proposition}[\cite{BM}]\label{prop: etale}
     $\G_\LL$ is a second countable, locally compact, Hausdorff, and étale groupoid for any semigroup $\LL\subset\Z^D$.
    \end{proposition}

It is known that $\G_\LL$ provides a realization of $C^*_r(\LL)$ as a groupoid $C^*$-algebra:

    \begin{proposition}\label{prop: groupoid}
The following isomorphism of $C^*$-algebras holds
    $$C_r^*(\LL)\;\simeq\;C^*_r(\G_\LL)$$
    where $C^*_r(\G_\LL)$ stands for the reduced groupoid $C^*$-algebra of $\G_\LL.$
\end{proposition}
\begin{proof}
The semigroup algebra $C_r^*(\LL)$ is by definition the same as the Wiener-Hopf algebra $\mathcal{W}(\LL)$ associated to $\LL$ as a subsemigroup of $\Z^D$, which can be characterized as the groupoid algebra of a transformation groupoid $X\rtimes \LL$ \cite[Theorem 5.5]{RS}. Here $X$ is a compact $\LL$-space that is essentially unique up to $\LL$-equivariant homeomorphism. It is easy to check that $\Xi_\LL$ endowed with the semigroup action $$\LL\times \Xi_\LL\ni (n,\Ss)\;\mapsto\; \Ss-n$$ satisfies the properties $A_1, A_2$ and $A_3$ listed in \cite[Section 5]{RS} which characterize that space and $\G_\LL = \Xi_\LL \rtimes \LL$ is a transformation groupoid. Therefore, the result follows from \cite[Theorem 5.5]{RS}.
\end{proof}

The isomorphism in Proposition \ref{prop: groupoid} can be made explicitly. In fact, for each $\lambda\in \LL$ consider the function $S_\lambda\colon\G_\LL\to \C$ defined as
$$S_\lambda(\Ss, x)\;=\;\delta_{\lambda,-x}$$
It is clear that $S_\lambda\in C_r^*(\G_\LL)$ since is the indicator function of the clopen subset $\{ (\Ss,-\lambda)\,|\,-\lambda\in \Ss\}\subset \G_\LL$. Moreover, it has the explicit adjoint $S_\lambda^*(\Ss,x)=\delta_{\lambda,x}$. Since the orbit of $\LL$ is dense in $\Xi_\LL$, then the left regular representation $\pi_\LL\colon C^*_r(\G_\LL)\to \mathscr{B}(\ell^2(\LL))$ is faithful \cite[Ex 5.3.3.]{Wil}. One has the relations $$\pi_\LL(S_\lambda)\;=\;V_\lambda,\qquad \pi_\LL(S_\lambda^*)\;=\;V_\lambda^*$$
thus demonstrating $C^*_r(\LL)\subset \pi_\LL(C^*_r(\G_\LL))$. The reverse inclusion also holds and thus $\pi_\LL(C_r^*(\G_\LL))= C^*_r(\LL)$ \cite[Theorem 5.5]{RS}. In particular, the commutative algebra $C(\Xi_\LL)$ can be identified with the  sub-$C^*$-algebra of $C_r^*(\LL)$ given by 
\begin{equation}\label{Eq:Idempotents}
    \mathfrak{F}_\LL\;:=\;C^*\big\{ V_l^*V_l\;|\; l\in \LL\big\}
\end{equation}
More precisely,  
$\mathfrak{F}_\LL$ is a commutative unital $C^*$-algebra so that its Gelfand spectrum agrees with $\Xi_\LL,$ \ie, 
$\mathfrak{F}_\LL\simeq C(\Xi_\LL)$ \cite{Li}.

\medskip

The partial action $\alpha$ of $\Z^D$ on $\Xi_\LL$ is the restriction of the $\Z^D$-action $\tilde{\alpha}_n(\Ss)=\Ss-n$ on $\tilde{\Xi}_\LL$. The inclusion $i\colon \Xi_\LL\subset \tilde{\Xi}_\LL$ is of course equivariant $\tilde{\alpha}_n(i(\Ss))=i(\alpha_n(\Ss))$ for any $\Ss\in U_n.$ One of the implications of this extension is that the dynamical system $(\tilde{\Xi}_\LL,\tilde{\alpha},\Z^D)$ captures important information of $C_r^*(\LL)$, as explained below:
\begin{theorem} [\cite{LAC}]\label{Teo morita}
    There is a full projection $p$ in the crossed product  $ C_0(\tilde{\Xi}_\LL)\rtimes_{\tilde{\alpha}}\Z^D$ such that
    $$C_r^*(\LL)\;\simeq\; p\big(C_0(\tilde{\Xi}_\LL)\rtimes_{\tilde{\alpha}}\Z^D\big)p$$
    Consequently, the $C^*$-algebras $C_r^*(\LL)$ and $C_0(\tilde{\Xi}_\LL)\rtimes_{\tilde{\alpha}}\Z^D$ are Morita equivalent.
\end{theorem}
\noindent Indeed, $p\in C_0(\tilde{\Xi}_\LL)$ is just the indicator function of the clopen subset $\Xi_\LL$. 
\begin{remark}
    Since the $K$-groups of a $C^*$-algebra are invariants under Morita equivalence, the $K$-theory of $C^*_r(\LL_{\bf v})$ can therefore in principle be computed via the Pimsner-Voiculescu exact sequence \cite{Pim}. 
    $\hfill \blacktriangleleft $
\end{remark}

\smallskip
Thanks to Proposition \ref{prop: groupoid}, there is a faithful conditional expectation map $E\colon C_r^*(\LL)\to C(\Xi_\LL)$. There is a natural semigroup action of $\LL$ on $\Xi_\LL$ defined via $$\LL\times \Xi_\LL\ni(l,\mathcal{S})\;\mapsto \;\mathcal{S}-l\in \Xi_\LL$$
 Notice that this action is well defined since $l\in \LL\subset  \mathcal{S}$ for any $\mathcal{S}\in \Xi_\LL.$ A subset $A\subset \Xi_\LL$ is invariant if $A-l=A$ for any $l\in \LL$. Since $\Xi_\LL$ is the unit space of $\G_\LL$, the latter is equivalent to saying that $\mathfrak{r}^{-1}(A)=\mathfrak{s}^{-1}(A)$, where we recall that $\mathfrak{r}$ and $\mathfrak{s}$ are the range and source maps, respectively. A Borel measure $\mu$ on $\Xi_\LL$ is said to be invariant if for any Borel measurable set $A\subset \Xi_\LL$ it satisfies  $\mu(A-l)=\mu(A)$. In this case, it is straightforward to verify that any invariant measure is also invariant in the groupoid sense \cite{RenaultBook}. An invariant probability measure on $\Xi_\LL$ is ergodic if every invariant set $A\subset \Xi_\LL$ satisfies $\mu(A)\in \{0,1\}.$  

\smallskip

As in the previous paragraph, throughout this work, we will use the groupoid and semigroup structure on $C^*_r(\LL)$, as both provide valuable insights into this $C^*$-algebra. In this way, for a closed/open invariant subset $A$ of $\Xi_\LL$ we adopt the notation $\LL|_A\equiv \G_\LL|_A$ where the latter is the \emph{reduction groupoid}, \ie
$$\G_\LL|_A\;:=\;\mathfrak{s}^{-1}(A)\cap\mathfrak{r}^{-1}(A)$$
This provides a splitting of $\Xi_\LL$ and, consequently, an exact sequence involving $C^*_r(\LL).$
\begin{proposition}[{\cite[Proposition 5.2]{Wil}}]\label{Prop:ES}
    Let $A$ be a closed invariant subset of $\Xi_\LL$ and $A^c$ its open complement. Then $C^*_r(\LL|_{A^c})$ is a closed ideal of $C^*_r(\LL)$ and there is a surjective $*$-homomorphism $\mathfrak{e}\colon C^*_r(\LL)\to C^*_r(\LL|_A)$ such that
    $$0\to C^*_r(\LL|_{A^c})\to C^*_r(\LL)\stackrel{\mathfrak{e}}\to C^*_r(\LL|_A)\to 0$$
\end{proposition}
\begin{remark}
    In the previous Proposition, we used the fact that $\G_\LL$ is a topological amenable groupoid since it is a locally compact subgroupoid of the amenable groupoid $\tilde{\Xi}_\LL\rtimes \Z^D$ \cite[Proposition 9.77]{Wil}.
    \hfill $\blacktriangleleft$
\end{remark}

\section{Measures on the hull of cone semigroups}
\label{sec:measures}
In this section, we describe the transversal hull of a cone semigroup and provide the proof of Theorem \ref{teo measures}.

\medskip

As it was pointed out in the introduction, let ${\bf v}:=\{ v_1,\dots, v_d\}$ be a set of normalized linearly independent vectors in $\R^D$ with $D\geq d.$ If $I\subset \{1,\dots,d\}$  we shall use the notation ${\bf v}_I:=\{ v_i\}_{i\in I}$ and set ${\bf v}_\emptyset=\emptyset$. Moreover, we also write ${\bf v}\setminus i\equiv {\bf v}_{\{1,\dots,d\}\setminus\{i\}}$. Associated with ${\bf v}$ there is the linear transformation $A_{\bf v}\colon \R^D\to \R^d$ with rows ${\bf v}$  acting on $z\in\R^D$ via
$$A_{\bf v} z\;=\;\sum_{i=1}^d \big(v_i\cdot z\big) e_i \in \R^d,$$
here $e_i$ is the standard basis of $\R^d.$ Consider a cone subsemigroup $\LL_{\bf v}$ of $\Z^D$ according to \eqref{eq: cone semigroup} and denote its transversal hull by $\Xi_{\bf v}\equiv \Xi_{\LL_{\bf v}}$.  The computation of this space requires introducing the image of $\LL_{\bf v}$ under the linear map $A_{\bf v}$, which defines a countable additive subsemigroup of $\R^d_+=[0,+\infty)^d$. Since it is not necessarily closed, we shall denote its closure as 
\begin{equation}\label{eq: decomposition semigroup}
    \mathcal{X}_{\bf v}\;:= \;\overline{A_{\bf v}(\LL_{\bf v})}\subseteq \R^d_+.
\end{equation}

\begin{definition}\label{def: rational dependence}
    We say that ${\bf v}$ is rational (R) if $A_{\bf v}(\LL_{\bf v})$ is a closed, discrete, and finitely generated subgroup of $\R^d_+$. Otherwise, we refer to ${\bf v}$ as irrational (I). If in particular $\X_{\bf v}=\R^d_+$, we shall say that ${\bf v}$ is completely irrational (CI). Furthermore, ${\bf v}$ has the property RCI if for any non-empty proper subset $I\subset \{1,\dots,d\}$  the restriction ${\bf v}_I$ is either R or CI. 
\end{definition}

\begin{remark}
   Observe that ${\bf v}$ is R if the semigroup $\LL_{\bf v}$ is finitely generated. For ${\bf v}$ to be CI, it is sufficient that all entries of the matrix $A_{\bf v}$ are linearly independent over $\mathbb{Q}$ and $D>d.$ Hence this is the generic case which holds for almost all ${\bf v}$. 
   \hfill $\blacktriangleleft$
\end{remark}
\medskip

In order to move forward in the description of $\Xi_{\bf v}$, for a subset $J$ of $\{1,\dots, d\}$ and $x\in \R^d$ consider the subsets of $\Z^D$
\begin{equation}\label{eq: j semigroup}
   \LL_{{\bf v},x}^J\;:=\;\big\{n\in \Z^D\;|\;v_k\cdot n+x_k>0\;\text{if}\;k\in J\;\text{and}\; v_k\cdot n+x_k\geq 0\;\text{if} \;k\notin J\big\}
\end{equation}
Note that $\LL_{{\bf v},x}^\emptyset=\LL_{\bf v}-m$,  whenever $x=A_{\bf v}m$ with $m \in \Z^D$.  Some of the sets in \eqref{eq: j semigroup} are contained in the transversal hull of $\LL_{\bf v}$: 
\begin{proposition}\label{prop: containing}
   Assume the RCI property on ${\bf v}$. Then the transversal hull of a cone semigroup satisfies 
    $$\big\{ \LL_{{\bf v},x}^\emptyset\,:\,x\in \mathcal{X}_{\bf v}\big\} \cup \bigcup_{i=1}^d \Xi_{{\bf v}\setminus i}\;\subset \; \Xi_{\bf v}$$
    where we recall that ${\bf v}\setminus i\equiv {\bf v}_{\{1,\dots,d\}\setminus\{i\}}$ and  $\Xi_\emptyset =\Xi_0=\{\Z^D\}$.
\end{proposition}
\begin{proof}
Let us first check that $\LL_{{\bf v},x}^\emptyset$ are elements of $\Xi_{\bf v}$ for any $x:=(x_k)_{k=1}^d\in \mathcal{X}_{\bf v}$. If $x\in A_{\bf v}(\LL_{\bf v})$ then $x=A_{\bf v}n$ for some $n\in \LL_{\bf v},$ and one gets the relation  $\LL_{{\bf v},x}^\emptyset=\LL_{\bf v}-n\in \mathscr{O}(\LL_{\bf v})$. Otherwise, for $x\in  \mathcal{X}_{\bf v} \setminus A_{\bf v}(\LL_{\bf v})$ there exists a sequence $n(j)\in \LL_{\bf v}$ such that $A_{\bf v}n(j)\to x$ with monotone components satisfying  $v_k\cdot n(j)\geq x_k.$ As a consequence of Lemma \ref{Lemma J}, one gets the convergence $\LL-n(j)\to \LL_{{\bf v},x}^\emptyset.$

To establish the inclusion of the other component let $\Ss=\LL_{{\bf v}\setminus i}-n\in \mathscr{O}(\LL_{{\bf v}\setminus i})$  with $n\in \LL_{{\bf v}\setminus i}$.  Pick a sequence $n(j)\in \LL_{\bf v}$ as in Lemma \ref{lemma convergence} such that $A_{{\bf v}\setminus i} n(j) \to A_{{\bf v}\setminus i}n$  and $v_i\cdot n(j)\to +\infty$. Then it is a consequence of Lemma~\ref{Lemma J} that $\LL_{\bf v}-n(j)\to \LL_{{\bf v}\setminus i}-n$ in the Fell topology. The  above together with the fact that $\Xi_{\bf v}$ is closed provide the inclusion $\Xi_{{\bf v}\setminus i}\subset \Xi_{\bf v}.$


\end{proof}
An induction on Proposition \ref{prop: containing} shows that $\LL_{{\bf v}_I}\in \Xi_{\bf v}$ for all $I\in \mathtt{P}(\{1,\dots d\})$, where $\LL_{{\bf v}_I}=\bigcap_{i\in I}\LL_{v_i}$ with the convention $\LL_{{\bf v}_\emptyset}=\Z^D.$  This in particular verifies the inclusion $$\bigsqcup_{I\in \mathtt{P}(\{1,\dots,d\})} \big\{ \LL_{{\bf v}_I,x}^\emptyset\,:\,x\in \X_{{\bf v}_I}\big\}\;\subset \; \Xi_{\bf v}$$
Depending on the nature of ${\bf v},$ the hull $\Xi_{\bf v}$ can also contain patterns of the form $\LL_{{\bf v},x}^J$ for non-trivial $J$. To label all possible cases where this may happen, for a non-empty subset $J\subset \{1,\dots,d\},$ define a dense subset of $\X_{\bf v}$ by
$$\mathtt{X}^J_{\bf v}\;:=\;\big\{ x\in \X_{\bf v}\;|\;\forall\,k\in J\;\exists\, n(k)\in \Z^D\;\text{such that}\; x_k=v_k\cdot n(k)\neq 0\;\big\}$$
Define also $\mathtt{X}^\emptyset_{\bf v}=\X_{\bf v}.$ It is clear that $A_{\bf v}(\LL_{\bf v})\subset \mathtt{X}_{\bf v}^{\{ 1,\dots,d\}}$ and $\mathtt{X}_{\bf v}^J\subset \mathtt{X}_{\bf v}^{J'}$ whenever $J\subset J'.$ We say that $J$ is maximal for $x\in \X_{\bf v}$ if $x\in \mathtt{X}_{\bf v}^{J}$ and $x\notin \mathtt{X}_{\bf v}^{J'}$ for any $J\subset J'.$
\begin{proposition}\label{prop: containing 2}Under the assumptions of Proposition \ref{prop: containing}, it holds that
  \begin{equation}\label{eq: hull CRI}
      \Xi_{\bf v}\; \subset \;\bigcup_{J\in \mathtt{P}(\{1,\dots,d\})}\big\{ \LL^J_{{\bf v},x}\,:\,x\in \mathtt{X}^J_{\bf v}\big\}\cup \bigcup_{i=1}^d \Xi_{{\bf v}\setminus i}
  \end{equation}
\end{proposition}
\begin{proof}
For $\mathcal{S}\in \Xi_{\bf v}$ let $n(j)$ be a sequence in $\LL_{\bf v}$ such that $\LL_{\bf v}-n(j)\to \mathcal{S}.$ For each $k\in \{1,\dots,d\}$ the sequence $\{v_k\cdot n(j)\}_{j\in \N}$ is bounded from below and therefore has a subsequence converging either to a finite number $x_k \in \R$ or to $+\infty$. By going over to a subsequence, we can therefore assume that $v_k\cdot n(j)$ converges for each $k$ and further, that it is either a strictly increasing or a non-increasing sequence. Denote by $J_+\subset \{1,\dots,d\}$ those values of $k$ for which $v_k\cdot n(j)$ is non-increasing and converges to a finite value as well as by $J_-$ those values of $k$ for which $v_k\cdot n(j)$ is strictly increasing and converges to a finite value. For $k$ in the complement $J_\infty = \{1,\dots,d\}\setminus (J_+\cup J_-)$ we can assume $v_k\cdot n(j)$ converges increasingly to $+\infty$.  As one can write $\LL_{\bf v}-n(j)=\LL_{{\bf v}- x(j)}$ for $x(j)=A_{\bf v} n(j)$, then the Lemma \ref{Lemma J} shows that $\LL_{\bf v}-n(j)$ converges in the Fell topology to the set $\LL_{{\bf v},x}^{J_+,J_-}$ with $x=\lim_{j\to \infty} A_{\bf v} n(j)\in \overline{\mathbb{R}}^d$ consisting of all points $n\in\mathbb{Z}^D$ such that
\begin{align*}
v_k \cdot n + x_k &\;\geq\; 0, \qquad \forall k \in J_+,\\
v_k \cdot n + x_k &\;> \;0, \qquad \forall k \in J_-
\end{align*}
Observe that if $J_\infty=\emptyset$ then $\Ss=\LL^{J_-}_{{\bf v},x}$ according to \ref{eq: j semigroup}. To see that one can restrict to $x\in \mathtt{X}^J_{\bf v}$ in the union \eqref{eq: hull CRI} one just needs to note that if $x \notin \mathtt{X}^{\{k\}}_{\bf v}$ then the distinction between strict and non-strict inequality is vacuous for the $k$-component, hence there always exists a proper subset $J'\subset J$ with $\LL^{J}_{{\bf v},x}=\LL^{J'}_{{\bf v},x}$ and $x\in \mathtt{X}^{J'}_{\bf v}$. The right-hand side of \eqref{eq: hull CRI}, therefore actually contains any set $\LL^J_{{\bf v},x}$ whenever $x\in \X_{\bf v}$ is a limit as constructed above. On the other hand, if $J_\infty\neq \emptyset$ then pick some $i\in J_\infty$ and consider the sub-tuple ${\bf v}\setminus i = (v_j)_{j\in \{1,...,d\}\setminus \{i\}}$ and the sequence $(\LL_{{\bf v}\setminus i}-n(j))_{j\in \N}$ in $\Xi_{{\bf v}\setminus i}$ for the same sequence $n(j)$. Applying Lemma~\ref{Lemma J} to that sequence one find that it converges to the same limit $\LL^{J_+,J_-}_{{\bf v},x}$ in the Fell topology, thereby showing that the limit point is already contained in $\Xi_{{\bf v}\setminus i}$.
\end{proof}

\begin{proposition}
\label{prop: CI}
If ${\bf v}$ is CI then
$$  \Xi_{\bf v}\; =\;\bigcup_{J\in \mathtt{P}(\{1,\dots,d\})}\big\{ \LL^J_{{\bf v},x}\;:\; x\in \mathtt{X}^J_{{\bf v}}\big\}\cup \bigcup_{i=1}^d \Xi_{{\bf v}\setminus i}$$
and thus by induction
\begin{equation}\label{eq: CI_hull}  \Xi_{\bf v}\; =\;\bigcup_{\substack{I\subset \{1,...,d\}\\ I\neq \emptyset}}\bigcup_{J\subset I}\big\{ \LL^J_{{\bf v}_I,x}\;:\; x\in \mathtt{X}^J_{{\bf v}_I}\big\}\;\cup \; \{\Z^D\}.
\end{equation}
\end{proposition}
\begin{proof}
Due to Proposition~\ref{prop: containing 2} we merely need to show that every set of the form $\LL_{{\bf v},x}^J$ is an element of $\Xi_{\bf v}$.

If $\X_{\bf v}=\R^d_+$ then one can for each point $x\in \mathtt{X}_{\bf v}^{J}$ find a sequence $\{x(j)\}_{j\in \N}$ in $A_{\bf v}(\LL_{{\bf v}})$ which converges to $x$, is increasing in the components in $J$ and non-increasing in the remaining components in $\{1,\dots,d\}\setminus J$. Choosing preimages $\{n(j)\}_{j\in \N}$ under $A_{\bf v}$ the Lemma \ref{Lemma J} verifies that $\LL_{{\bf v}}-n(j)$ converges in the Fell topology to $\LL^J_{{\bf v},x}$.
\end{proof}

The other special case is the rational case:
\begin{proposition}\label{prop: RD}
If ${\bf v}$  is R the transversal hull $\Xi_{\bf v}$ is countable and  satisfies
    \begin{equation}\label{eq: RD}
        \Xi_{\bf v}\;=\;\bigsqcup_{I\in \mathtt{P}(\{1,\dots,d\})}\mathscr{O}(\LL_{{\bf v}_I})
    \end{equation}

\end{proposition}
\begin{proof}
Due to the Propositions~\ref{prop: containing} and \ref{prop: containing 2} it is enough to show that for each $I\subset \{1,...,d\}$ one has
$$\{ \LL^J_{{\bf v}_I,x}\,:\,x\in \mathtt{X}^J_{{\bf v}_I}\}\;=\;\mathscr{O}(\LL_{{\bf v}_I})\;=\;\{ \LL_{{\bf v}_I,x}^\emptyset\,:\,x\in \X_{{\bf v}_I}\}.$$
The second equality reproduces exactly the definition of $\mathscr{O}(\LL_{{\bf v}_I})$. Observe that due to rationality one has for each $I$ and $k\in I$ a minimal period $c_{I,k}>0$ such that $v_k\cdot \LL_{{\bf v}_I}=c_{I,k}\Z_+$. Therefore, one can write any $\LL^J_{{\bf v}_I,x}$ in the form $\LL^\emptyset_{{\bf v}_I,y}$ for $y\in \X_{{\bf v}_I}$ given by
$$y = x - \sum_{k\in J} c_{I,k} e_k$$
with $e_k$ the unit vectors of $\R^I$. Note that $y$ lies in the semigroup $\X_{{\bf v}_I}$ since the assumption $x\in \mathtt{X}^J_{{\bf v}_I}$ includes $x_k\neq 0$ and therefore $x_k\geq c_{I,k}$.
\end{proof}
Note that the expressions \eqref{eq: CI_hull} and \eqref{eq: RD} coincide if ${\bf v}$ is $R$, hence the former can also be used in the rational case. The $\Z^D$-hulls can be computed similarly:
\begin{corollary}
If ${\bf v}$ is RCI then
\begin{equation}
\label{eq: ZD_hull} \tilde{\Xi}_{\bf v}\; =\;\bigcup_{\substack{I\subset \{1,...,d\}\\ I\neq \emptyset}}\bigcup_{J\subset I}\big\{ \LL^J_{{\bf v}_I,x}\;:\; x\in \tilde{\mathtt{X}}^J_{{\bf v}_I}\big\}\;\cup \; \{\Z^D\}
\end{equation}
with $$\tilde{\mathtt{X}}^J_{{\bf v}_I}\;=\;\big\{ x\in \R^I\;|\;\forall\,k\in J\;\exists\, n(k)\in \Z^D\;\text{such that}\; x_k=v_k\cdot n(k)\;\big\}.$$
\end{corollary}
\begin{proof}
By Proposition~\ref{prop: CI} and Proposition~\ref{prop: RD} the right-hand side of \eqref{eq: ZD_hull} is nothing but the $\Z^D$-orbit $\mathscr{O}_{\Z^D}(\Xi_{\bf v})$ and we have $\tilde{\Xi}_{\bf v}\subset \mathscr{O}_{\Z^D}(\Xi_{\bf v})$ by definition of $\tilde{\Xi}_{\bf v}$. For the reverse inclusion let $(n(j))_{j\in \N}$ be a sequence in $\Z^D$ such that $\LL_{\bf v}-n(j)$ converges in the Fell topology to some $\Ss\in \mathcal{C}(\Z^D)\setminus \emptyset$. We need to prove $\Ss\in \mathscr{O}_{\Z^D}(\Xi_{\bf v})$. Note that for any $k\in \{1,...,d\}$ one must have $\sup_{j\in \N} v_k\cdot n(j) < \infty$ since otherwise $\Ss$ would be empty by a similar argument as in Lemma~\ref{Lemma J}. Therefore, there exists some $\tilde{n}\in \Z^D$ such that $\LL_{\bf v}-(n(j) -\tilde{n})$ is a convergent sequence in $\mathscr{O}(\LL_{\bf v})$, which shows $\Ss+\tilde{n}\in \Xi_{\bf v}$ and thus $\Ss\in \mathscr{O}_{\Z^D}(\Xi_{\bf v})$.
\end{proof}

In order to get for any ${\bf v}$ a convenient disjoint decomposition of $\Xi_{\bf v}$ in invariant subsets, as in \eqref{eq: RD}, let us denote by 
\begin{equation}\label{eq: clopen}
\mathcal{C}_{{\bf v}_I}\;:=\;  \Big(\bigcup_{J\in \mathtt{P}(I)}\big\{ \LL^J_{{\bf v}_I,x}\;:\; x\in \mathtt{X}^J_{{\bf v}_I}\big\}\Big)\cap \Xi_{{\bf v}_I}\;=\; \Xi_{{\bf v}_I} \setminus \left(\bigcup_{i \in I}\Xi_{{\bf v}_I\setminus i}\right)
\end{equation}
 with  $\mathcal{C}_{{\bf v}_\emptyset}=\{\Z^D\}$. The equality of the two variants of the definition follows from the fact that the union on the right-hand side of \eqref{eq: hull CRI} is disjoint.  As an intersection of invariant sets $\mathcal{C}_{{\bf v}_I}$ is a non-trivial invariant subset of $\Xi_{\bf v}$ and it is open in the relative topology of $\Xi_{{\bf v}_I}$. 
 We arrive at the following description of $\Xi_{\bf v}$:
\begin{proposition}\label{prop: general case} Let the RCI property be valid on ${\bf v}$. Then the transversal hull admits the disjoint decomposition
  \begin{equation}\label{eq: decomposition cantor}
    \Xi_{\bf v}\;=\; \bigsqcup_{I\in \mathtt{P}(\{1,\dots,d\})}\mathcal{C}_{{\bf v}_I}
\end{equation}
into invariant subsets $\mathcal{C}_{{\bf v}_I}$. This induces  a filtration of $\Xi_{\bf v}$ by closed invariant subsets
\begin{equation}\label{eq: filtration 1}
\{\Z^D\}=\Xi_0\;\subset \;\Xi_1\;\subset\; \cdots \;\subset\; \Xi_{d-1}\;\subset \;\Xi_{d}=\Xi_{\bf v}
\end{equation}
with $\Xi_r=\Xi_{\bf v}\setminus \bigsqcup_{|I|>r}\mathcal{C}_{{\bf v}_I} =\bigcup_{|I|=r}\Xi_{{\bf v}_I}$. In particular, one has the relation
\begin{equation}\label{eq: boundary set}
    \Xi_{r}\setminus \Xi_{r-1}\;=\;\bigsqcup_{|I|=r}\mathcal{C}_{{\bf v}_I}
\end{equation}
where each $\mathcal{C}_{{\bf v}_I}$ is open in $\Xi_r\setminus\Xi_{r-1}$ with the subspace topology.  
\end{proposition}

\medskip 

The $\Z^D$-hull $\tilde{\Xi}_{\bf v}$ of $\LL_{\bf v}$ can also be filtered similarly:

\begin{corollary}\label{coro: zd hull}
   Under the RCI assumption on ${\bf v},$  the $\Z^D$-hull of a cone semigroup admits a disjoint decomposition 
   \begin{equation}\label{eq: decomposition cantor2}
    \tilde{\Xi}_{\bf v}\;=\; \bigsqcup_{I\in \mathtt{P}(\{1,\dots,d\})}\tilde{\mathcal{C}}_{{\bf v}_I}
\end{equation}
in $\Z^D$-invariant subsets. There is also a filtration by closed invariant subsets 
$$\{\Z^D\}=\tilde{\Xi}_0\subset \tilde{\Xi}_1\subset\dots\subset \tilde{\Xi}_{d-1}\subset \tilde{\Xi}_{d}=\tilde{\Xi}_{\bf v}$$
with  $\tilde{\Xi}_r=\tilde{\Xi}_{\bf v}\setminus \bigsqcup_{|I|>r}\tilde{\mathcal{C}}_{{\bf v}_I}$.
\end{corollary}

Notice that the space $\X_{{\bf v}}$ admits a  semigroup action of $\LL_{\bf v}$ defined via
$$(x,n)\;\mapsto x+A_{{\bf v}}n\qquad n\in \LL_{\bf v},\;x\in \X_{{\bf v}}.$$ 
Assume now that ${\bf v}$ has the property RCI. There is then, up to scaling factor, a unique $\LL_{\bf v}$-invariant Radon measure given by either the counting measure (if $\X_{{\bf v}}$ is discrete) or the restriction of the Lebesgue measure (if $\X_{{\bf v}}=\R_+^d$). For the latter, just note that any invariant measure is invariant under a dense subsemigroup of $\R^d_+$ and thus under translations by all of $\R^d_+$ by regularity.

\medskip 

Now we are ready to compute the vector space $\mathfrak{M}(\Xi_{\bf v})$ of all invariant measures on $\Xi_{\bf v}$ which are Radon measures supported on the boundaries $\Xi_r\setminus \Xi_{r-1}$. We start with the following preparatory Lemma:

\begin{lemma}\label{coro: surjective} Let ${\bf v}$ has the property RCI. Then there is a proper, continuous, surjective map $\Gamma_I\colon \mathcal{C}_{{\bf v}_I}\to \X_{{\bf v}_I}$ such that 
    \begin{equation}\label{eq: equivariant}
        \Gamma_I(\mathcal{S}-n)\;=\; \Gamma_I(\mathcal{S})+A_{{\bf v}_I}n
    \end{equation}
for all $\mathcal{S}\in \mathcal{C}_{{\bf v}_I}$ and $n\in \LL_{{\bf v}_I}$. 
\end{lemma}

\begin{proof}
For the sake of notational simplicity, let us remove the $I$-dependence throughout this proof. There are two cases, either $\X_{{\bf v}}$ is discrete or $\X_{\bf v}=\R^{d}_+$. If  $\X_{{\bf v}}$ is discrete define $\Gamma$ as the bijective correspondence
$$\mathcal{C}_{{\bf v}}=\mathscr{O}(\LL_{{\bf v}})\ni \LL_{{\bf v}}-n\;\mapsto \; A_{{\bf v}}n.$$
It is equivariant, and since both spaces are discrete, it is an equivariant homeomorphism. 

In the case $\X_{\bf v}=\R^{d}_+$ define $\Gamma(\mathcal{S})=x\in \R^d_+$ where the coordinates of the vector $x$ are given by
$$x_k\;:=\;-\inf_{n\in \mathcal{S}} v_k\cdot n\;\geq\;0$$
Since $A_{\bf v}(\LL_{{\bf v},x}^J)$ is dense in $x+ \X_{\bf v}$ for any $J$ the map $\Gamma$ is well-defined on $\mathcal{C}_{\bf v}$ and already uniquely determined by $\Gamma(\LL_{{\bf v},x}^J)=x$. Observe also that $\Gamma$ is surjective by Proposition \ref{prop: containing} and clearly satisfies the equivariance condition given in \eqref{eq: equivariant}. To check sequential continuity we consider a convergent sequence in $\mathcal{C}_{{\bf v}}$, which always takes the form $\LL^{J(j)}_{{\bf v},x(j)} \to \LL^J_{{\bf v},x}$ with $x(j) \in \R^d_+$ and $j \in \mathbb{N}$.   Since $A_{\bf v}(\Z^D)$ is dense in $\R^d$  there exist for any $\epsilon> 0$ some $n'\in \Z^D$ 
 such that 
$$\epsilon\geq v_k\cdot n'+ x_k> 0$$ for all $1\leq k\leq d$, hence $n'\in \LL^J_{{\bf v},x}$. By convergence in the Fell topology one must also have  $n'\in \LL^{J(j)}_{{\bf v}, x(j)}$ for all large enough $j$ and thus
\begin{equation}\label{eq: ineq 1}
  \epsilon-x_k+x_k(j)\;\geq \; v_k\cdot n'+x_k(j)\;\geq\; 0. 
\end{equation}
We have  $x_k(j) \geq 0$ for all $k=1,...,d$. Fix now some $l\in \{1,...,d\}$ and use the density again to pick some $n''\in \Z^D$ such that 
$$v_k\cdot n''+x_k(j)\;\geq\;v_k\cdot n''\;\geq\;0\qquad \forall\, l\neq k$$
$$v_l\cdot n''+x_l\;<\;0\;<\;v_l\cdot n''+x_l+\epsilon.$$
Since $n''\notin \LL_{{\bf v},x}^J$ one must also have $n''\notin \LL_{{\bf v},x(j)}^{J(j)}$ for all large enough $j$. As a conclusion, 
$$v_l\cdot n''+x_l(j)\;\leq\;0\;<\;v_l\cdot n''+x_l+\epsilon$$
Together with \eqref{eq: ineq 1} this verifies that $\lvert x_l-x_l(j)\rvert\leq \epsilon$ and therefore, since $l$ was arbitrary, the continuity of $\Gamma$. 

It remains to prove that $\Gamma$ is a proper map. Let $B \subset \R^d_+$ be a closed bounded set and put $K = \Gamma^{-1}(B)$. By continuity, $K$ is a closed subset of $\mathcal{C}_{{\bf v}}$ and it consists of precisely all elements $\LL^J_{{\bf v}, x}\in \mathcal{C}_{{\bf v}}$ for which $x\in B$. For compactness, it is enough to check that every sequence $\mathcal{S}_j:=\LL^{J(j)}_{{\bf v}, x(j)}$ with $x(j)\in B$ has a limit point in $K$. Due to compactness of $\Xi_{\bf v}$ it has a limit point $\mathcal{S}\in \Xi_{\bf v}.$ As in the proof of Proposition~\ref{prop: containing 2} one can go over to a monotonous subsequence and since $x(j)$ is uniformly bounded the same reasoning shows that the limit point must be of the form $\mathcal{S}=\LL_{{\bf v},x}^J\in \mathcal{C}_{\bf v}$ for some  $J$ and $x\in B$, thus $\mathcal{S}\in K$. In summary, $\Gamma$ is a proper, continuous, surjective, and equivariant map.

\end{proof}

Denote by $\mathfrak{R}(X)$ the vector space of $\LL$-invariant Radon measures on a topological space $X$ endowed with semigroup action by $\LL$. We now present the proof of our first main result.

\medskip
\noindent
{\bf Proof of the Theorem \ref{teo measures}:}
\begin{proof}
Since $\mathcal{C}_{{\bf v}_I}$ is open in $\Xi_r\setminus \Xi_{r-1}$ with $|I|=r$, the decomposition of $\Xi_{\bf v}$ in invariant subsets given in \eqref{eq: decomposition cantor} implies that
    $$\mathfrak{M}(\Xi_{\bf v})\;=\;\bigoplus _{I\in \mathtt{P}(\{1,\dots,d\})}\mathfrak{R}(\mathcal{C}_{{\bf v}_I})$$
    Thus, to conclude the proof, it is enough to show that $|\mathfrak{R}(\mathcal{C}_{{\bf v}_I})|=1$ for all $I$. Let $\mu$ be an invariant Radon measure on $\mathcal{C}_{{\bf v}_I}.$ Then the pushforward measure $\Gamma_{I_*}(\mu)$ is an invariant Radon measure on $\X_{{\bf v}_I}$ by Lemma \ref{coro: surjective}. We claim that $\mu\mapsto \Gamma_{I_*}(\mu)$ is an isomorphism of vector spaces and  $|\mathfrak{R}(\mathcal{C}_{{\bf v}_I})|=|\mathfrak{R}(\X_{{\bf v}_I})|=1$. 

We will now drop the subscript $I$. In the discrete case, the isomorphism is obvious since $\Gamma$ is an equivariant homeomorphism. Let us therefore assume the dense case $\X_{{\bf v}}=\R_+^{d}$. Let $\mu$ be a signed invariant Radon measure on $\mathcal{C}_{{\bf v}}$. For $J\in \mathtt{P}(\{1,...,d\})$ define $$A^J \;=\; \big\{\LL^J_{{\bf v},x}: x \in \mathtt{X}^J_{{\bf v}}\big\}, \qquad \mu^J(A)\;:=\; \mu(A \cap A^J),$$ which are invariant measurable sets and invariant Radon measures on $\mathcal{C}_{\bf v}$, respectively. Since the Lebesgue measure $\nu$ is the unique translation-invariant Radon measure, one must have $\Gamma_*\mu^J= c_J \nu$ for some constants. Note, however, that $\Gamma_*\mu^J$ is supported in $\mathtt{X}^J_{{\bf v}}$, which has Lebesgue measure $0$ for $J\neq \emptyset$, since it is contained in a countable union of $(d-1)$-dimensional hyperplanes. Therefore, $\Gamma_*\mu^J = 0$ for $J\neq \emptyset$. Since $\Gamma\rvert_{A^J}$ is injective, this implies $\mu^J=0$ and hence $A^J$ has $\mu$-measure $0$ as well. We conclude that $\mu = \mu^{\emptyset}$ because the sets $(A^J)_{J\in \mathtt{P}(\{1,...,d\})}$ cover $\mathcal{C}_{\bf v}$. If $\Gamma_*\mu=0=\Gamma_*\mu^\emptyset$ it follows that $\mu^\emptyset = 0$, since $\Gamma\rvert_{A^\emptyset}$ is also injective, and hence finally $\mu=0$, showing injectivity.

For surjectivity we construct an explicit inverse to $\Gamma_*$ given by the pull-back $\Gamma^*$, mapping a measure on $\X_{\bf v}$ to a set-valued function on $\mathcal{C}_{\bf v}$ defined as
\begin{equation}\label{eq: measure inverse}
    (\Gamma^*\nu)(A) \;=\; \nu(\Gamma(A))
\end{equation}
for each Borel set $A\subset \mathcal{C}_{\bf v}$. Let us prove that this pull-back measure is well-defined. First of all, $\Gamma(A)$ is measurable for any Borel set, since $\Gamma$ is a continuous map and satisfies the conditions of \cite[2.2.13]{FED}\footnote{$\mathcal{C}_{{\bf v}}$ is a Polish space since it is an open subset of $\Xi_{{\bf v}}$}. Let us note that $\Gamma$ is almost injective in the sense that there is a set $N\subset \X_{\bf v}$ of Lebesgue-measure $0$ such that $\Gamma\rvert_{\mathcal{C}_{\bf v} \setminus \Gamma^{-1}(N)}$ is injective, namely we can use $N=\bigcup_{J \in \mathtt{P}(\{1,...,d\})\setminus \emptyset}\mathtt{X}^J_{\bf v}$. It is then standard to prove that $\Gamma^*\nu$ is a countably additive measure whenever $\nu$ is absolutely continuous w.r.t. the Lebesgue measure. Indeed, $\Gamma^*\nu$ is then nothing but the push-forward of $\nu$ w.r.t. the proper map $(\Gamma\rvert_{\mathcal{C}_{\bf v} \setminus \Gamma^{-1}(N)})^{-1}$. Since $\Gamma$ maps compact sets to compact sets, $\Gamma^*\nu$ is also locally finite and, therefore, inner and outer regular (since $\mathcal{C}_{\bf v}$ is a separable locally compact space and hence $\sigma$-compact).

Finally, observe that the set of ergodic probability measures is contained in $\mathfrak{M}(\Xi_{\bf v})$, and since the basis is given by infinite measures except for the Dirac measure  $\mu_\emptyset$ on the single point $\{\Z^D\},$  then the unique ergodic probability measure on $\Xi_{\bf v}$ is $\mu_\emptyset.$ 
\end{proof}

\begin{remark}\label{rem: basis of measure}
    An explicit base for $\mathfrak{M}(\Xi_{\bf v})$ is given by the set of  measures $\{\mu_I\}_{I\in \mathtt{P}(\{1,\dots,d\})}$ with normalization chosen so that if $\X_{{\bf v}_I}$ is discrete then
$$\mu_I(A)\;:=\;{{\rm Vol}\big(\R^{I}/\tilde{\X}_{{\bf v}_I}\big)}\nu_I(\Gamma(A))$$
with $A\subset \mathcal{C}_{{\bf v}_I}$ and $\nu_I$ the counting measure on $\Z^{I}_+$. If  $\X_{{\bf v}_I}=\R^{I}_+$ then $\Gamma_*(\mu_I)=\nu_I$ with  $\nu_I$ the normalized Lebesgue measure on $\R^{I}_+$.
\hfill $\blacktriangleleft$
\end{remark}
\begin{remark}\label{rem: basis of measure zd}
    Let $\mathfrak{M}(\tilde{\Xi}_{\bf v})$ be the vector space of those $\Z^D$-invariant Borel measures on $\tilde{\Xi}_{\bf v}$ which are trivial extensions of Radon measures on some $\tilde{\Xi}_r\setminus \tilde{\Xi}_{r-1}$. Under the same assumptions as Theorem~\ref{teo measures} this space also has dimension $2^d$ and  a basis is provided by the set $\{ \tilde{\mu}_I\}_{I\in \mathtt{P}(\{1,\dots,d\})}$ where each $\tilde{\mu}_I$ is uniquely determined by the relation $\tilde{\mu}_I|_{\Xi_{\bf v}}=\mu_I.$
\hfill $\blacktriangleleft$
\end{remark}

\section{Chern numbers in the cone geometry} \label{Sec: chern}
This section is devoted to the proof of Theorem \ref{Theorem 2}. Here, we shall construct the correct nontrivial Chern cocycle on the boundary algebras $\mathfrak{I}_r$, defined by the trace per unit hypersurface induced by the measures computed in the previous section.

\medskip

Consider the \emph{cone} semigroup $C^*$-algebra $\A_{\bf v}:=C^*_r(\LL_{\bf v})$ with $\bf v$ displaying the RCI property. This admits the cofiltration
\begin{equation}
    \A_{\bf v}\;=\;\A_{d}\stackrel{\psi_d}\to \A_{d-1}\stackrel{\psi_{d-1}}\to \cdots \to \A_1\stackrel{\psi_1}\to \A_0\;=\;C^*_r(\Z^D)
\end{equation}
induced by the filtration \eqref{eq: filtration 1} of $\Xi_{\bf v}$. Here $\A_{r}:=C^*_r(\LL_{\bf v}|_{\Xi_r})$ and $\psi_r$ is the surjective $*$-homomorphism given by the restriction from $\Xi_r$ to $\Xi_{r-1}$.  The codimension-$r$ boundary algebra is then defined as $  \mathfrak{I}_r={\rm Ker}(\psi_r)$ and $\mathfrak{I}_r \simeq C^*_r(\LL_{\bf v}|_{\Xi_r\setminus\Xi_{r-1}})$, according to \ref{Prop:ES}. From the construction, we have the following short exact sequences
\begin{equation}\label{eq: seq 2}
0\to \mathfrak{I}_r\to \A_r\to \A_{r-1}\to 0
\end{equation}
Observe that $\mathfrak{I}_r$ can be decomposed in direct sum as
\begin{equation}\label{eq: edge algebra}
  \mathfrak{I}_r\;=\;\bigoplus_{|I|=r} \mathfrak{I}_I  
\end{equation}
where $\mathfrak{I}_I:=C^*_r(\LL_{\bf v}|_{\mathcal{C}_{{\bf v}_I}})$ is a closed ideal of $\A_r$ by Proposition \ref{prop: general case}.  Define the linear functional on $\mathfrak{I}_I$ given by 
\begin{equation}
    \Tt_I(f)\;:=\;\int_{\mathcal{C}_{{\bf v}_I}}E(f)(x){\rm d}\mu_I(x),\qquad f\in \mathfrak{I}_I
\end{equation}
where $\mu_I$ is the unique invariant Radon measure on $\mathcal{C}_{{\bf v}_I}$ defined in Remark \ref{rem: basis of measure}. As a consequence of the properties of $\mu_I$, this functional satisfies:
\begin{proposition}\label{prop: traces}
 $\Tt_I$ is a densely defined trace on $\mathfrak{I}_I$ which is faithful and lower semicontinuous.
\end{proposition}
 The real-space representation of $\Tt_I$  and its interpretation as a trace-per-surface-area will be discussed in Section \ref{sec: trace}.
\medskip

We are now in place to define the Chern cocycles on $\mathfrak{I}_r$. Consider  a tuple of normalized vectors ${\bf w}=( w_1,w_2,\dots,w_m)$ in $\R^D$ and define the $(m+1)$-linear functional
\begin{equation}
    {\rm Ch}_{I,{\bf w}}(f_0,f_1,\dots, f_m)\;:=\;\sum_{\rho\in S_m}(-1)^\rho\Tt_I\big(f_0\nabla_{w_{\rho(1)}}f_1\cdots \nabla_{w_{\rho(m)}}f_m\big)
\end{equation}
which is well-defined for any $(m+1)$-tuple of elements $f_k$ in a suitable dense subalgebra of $\mathfrak{I}_I$ consisting of elements in the domain of both $\nabla$ and $\Tt_I$. Here, $S_m$ is the symmetric group of $m$ elements, $(-1)^\rho$ stands for the sign of $\rho$ and the directional derivatives are defined according to
\begin{equation}\label{eq: deriva}
    \nabla_vf\;:=\; v\cdot \nabla f\;=\;\sum_{i=1}^Dv_i\nabla_if
\end{equation}
with $\nabla_if=\ii[\mathfrak{n}_i,f]$ and $\mathfrak{n}_i$ the position operator in the direction $i$ on $\ell^2(\LL_{\bf v})$. Thanks to the properties of the trace $\Tt_I$, it follows that ${\rm Ch}_{I,{\bf w}}$ defines a cyclic $m$-cocyle on $\mathfrak{I}_I$ \cite{Con}. For even $m$, this Chern cocycle pairs with the group $K_0(\mathfrak{I}_I)$ via
$$\langle [p]_0, [{\rm Ch}_{I, {\bf w}}]\rangle \;=\;\frac{1}{(m/2)!(-2\pi \imath)^{m/2}}\;{\rm Ch}_{I, {\bf w}}(p,\dots,p)$$
and for odd $m$ with $K_1(\mathfrak{I}_I)$ via 
$$\langle [u]_1, [{\rm Ch}_{I, {\bf w}}]\rangle \;=\;\frac{\imath^{(m+1)/2}}{m!!(-2)^m\pi^{(m+1)/2}}\,{\rm Ch}_{I, {\bf w}}(u^{-1}-1, u-1,u^{-1}-1,\dots,u-1).$$

These pairings are group homomorphisms with respect to $K_*(\mathfrak{I}_I)$ and depend only on the cohomology class of the $m$-cycle ${\rm Ch}_{I,{\bf w}}$ \cite{Con}.

\medskip
   Define the \emph{extended} cone algebra as the crossed product $\tilde{\A}_{\bf v}:=C_0(\tilde{\Xi}_{\bf v})\rtimes\Z^D$ where $\tilde{\Xi}_{\bf v}$ is provided in Corollary \ref{coro: zd hull}.  This algebra also admits a cofiltration
   \begin{equation}\label{eq: extended cofiltrarion}
    \tilde{\A}_{\bf v}=\tilde{\A}_{d}\stackrel{\tilde{\psi}_d}\to \tilde{\A}_{d-1}\stackrel{\tilde{\psi}_{d-1}}\to \cdots \to \tilde{\A}_1\stackrel{\tilde{\psi}_1}\to \tilde{\A}_0=C^*_r(\Z^D)
\end{equation}
with $\tilde{\A}_r=C_0(\tilde{\Xi}_r)\rtimes \Z^D$ and ideals $\tilde{\mathfrak{I}}_r:={\rm Ker}(\tilde{\psi}_r).$ Similarly, one has the decomposition $\tilde{\mathfrak{I}}_r=\bigoplus_{|I|=r}\tilde{\mathfrak{I}}_I$, where $\tilde{\mathfrak{I}}_I:=C_0(\tilde{\mathcal{C}}_{{\bf v}_I})\rtimes \Z^D$ are closed ideals of $\tilde{\A}_r$ with a unique tracial weight defined via
\begin{equation}
    \tilde{\Tt}_I(f)\;:=\;\int_{\tilde{\mathcal{C}}_{{\bf v}_I}}E(f)(x){\rm d}\tilde{\mu}_I(x),\qquad f\in \tilde{\mathfrak{I}}_I\,.
\end{equation}
Here $\tilde{\mu}_I$ is the measure in Remark \ref{rem: basis of measure zd}. Thus, the natural Chern cocycle on $\tilde{\mathfrak{I}}_I$ is
\begin{equation}
    \tilde{\rm Ch}_{I,{\bf w}}(f_0,f_1,\dots, f_m)\;:=\;\sum_{\rho\in S_m}(-1)^\rho\tilde{\Tt}_I\big(f_0\nabla_{w_{\rho(1)}}f_1\cdots \nabla_{w_{\rho(m)}}f_m\big)
\end{equation}

The next Proposition will complete the Proof of Theorem \ref{Theorem 2}:
\begin{proposition}\label{prop: numerical invariants}
 For any linearly independent tuple $(w_1,...,w_m)$ of vectors orthogonal to ${\bf v}_I$ there is an elements$ [\zeta]_i\in K_i(\mathfrak{I}_I)$, $i= m\, \mathrm { mod }\,2$, such that the pairing $\langle [\zeta]_i,[{\rm Ch}_{I,{\bf w}}]\rangle$ does not vanish.
\end{proposition}
Before we can prove this Proposition, we first need to introduce another variant of the transversal hull, which will define a \emph{smooth} cone algebra $\mathcal{A}_{\bf v}$. Consider the locally compact space
\begin{equation}\label{eq: Omega}
  \Omega_{\bf v}\;=\;\overline{\big\{ \R_+^d-x\;|\;x\in \R^d\big\}}\setminus \emptyset  
\end{equation}
where the closure is taken with respect to the Fell topology on $\mathscr{C}(\R^d).$ 
Translations in $\Z^D$ act via
$$(\Ss,n)\;\mapsto\; \beta_n(\Ss)\;=\;\Ss - A_{{\bf v}}n,\qquad \forall \, \Ss\in \Omega_{\bf v},n \in \Z^D.$$
In \cite[Proposition 3.2]{DET2} it is shown that $\Omega_{\bf v}$ has a decomposition as a CW-complex with cell decomposition
$$\Omega_{\bf v}\;\simeq\;\bigsqcup_{I\in \mathtt{P}(1,\dots,d)} \R^{I}$$
where $\R^{\emptyset}=\{*\}$ is a single point. This is only a bijection, not a homeomorphism, since the cells are glued non-trivially. Precisely, any element of $\Omega_{\bf v}$ can be written uniquely in the form $\Ss=(x+ \R_+^I)\times \R^{I^c}$ for some $I\subset \{1,...,d\}$ and $x\in \R^I$. Under the bijection the $\Z^D$-action $\beta$ turns into the \emph{affine} transformations
$$(x,n)\;\mapsto\; \beta_n(x)\;=\;x+A_{{\bf v}_I}n,\qquad \forall \,x\in \R^{I},n \in \Z^D.$$
The \emph{smooth} cone algebra $\mathcal{A}_{\bf v}$ shall be the crossed product algebra under this action, \ie
$$\mathcal{A}_{\bf v}\;:=\;C_0(\Omega_{\bf v})\rtimes_\beta\Z^D.$$
This algebra compares to $C_0(\tilde{\Xi}_{\bf v})\rtimes \Z^D$ in a similar way that the so-called smooth Toeplitz algebra \cite{JI} relates to the usual Toeplitz algebra.

The crossed product algebra $\mathcal{A}_{\bf v}$ can be described as the universal $C^*$-algebra generated by formal Fourier series
$$a \;=\; \sum_{n\in \Z^D} a_n u^n$$
with commuting unitaries $u^n=u_1^{n_1}\cdots u_D^{n_D}$ representing the generators of the $\Z^D$-action and coefficients $f_n\in C_0(\Omega_{\bf v})$ satisfying the commutation relation $u^nf_m u^{-n}=\beta_n(f_m)$. Moreover, one obtains a strongly continuous family of $*$-representations $\{\pi_\omega\}_{\omega\in\Omega_{\bf v}}$ of $\mathcal{A}_{\bf v}$ on $\ell^2(\Z^D)$ determined by the relation
$$\langle n|\pi_\omega(a)|m\rangle\;=\;a_{n-m}\big(\beta_{n-m}(\omega)\big).$$

Similarly to the unit-space $\tilde{\Xi}_{\bf v}$ the space $\Omega_{\bf v}$ also has a similar filtration by closed $\Z^D$-invariant subsets
\begin{equation}\label{eq: filtration smooth 2}
    \{ *\}=\Omega_0\subset \Omega_1\subset\dots \subset \Omega_{d-1}\subset \Omega_d =\Omega_{\bf v}
\end{equation}
where $\Omega_r:=\Omega_d\setminus \bigsqcup_{|I|>r}\R^{I}$. This in turn a cofiltration of $\mathcal{A}_{\bf v}$ 
\begin{equation}\label{eq: cofiltration smooth}
    \mathcal{A}_{\bf v}\;=\;\mathcal{A}_{d}\stackrel{\phi_d}\to \mathcal{A}_{d-1}\stackrel{\phi_{d-1}}\to \cdots \to \mathcal{A}_1\stackrel{\phi_1}\to \mathcal{A}_0\;\simeq\; C(\T^D)
\end{equation}
with $\phi_r$ surjective $*$-homomorphism and $\mathcal{A}_r=C_0(\Omega_r)\rtimes_{\beta}\Z^D$. The smooth boundary ideals here are
$\mathcal{I}_r:={\rm Ker}(\phi_r)$. It turns out that the following decomposition in direct sum
\begin{equation}\label{eq: smooth edge algebra}
\mathcal{I}_r\;=\;\bigoplus_{|I|=r}\mathcal{I}_I
\end{equation}
where each $\mathcal{I}_I\;:=\;C_0(\R^{I})\rtimes_{\beta}\Z^D$ is an ideal of $\mathcal{A}_r$. The canonical trace on $\mathcal{I}_I$ is given by 
\begin{equation}
        \mathcal{T}_I(f)\;:=\;\int_{\R^{I}}E(f)(x)\,{\rm d}\nu_I(x)
\end{equation}
where $\nu_I$ is the normalized Lebesgue measure on $\R^{I}$ and $E\colon \mathcal{A}_{\bf v}\to C_0(\Omega_{\bf v})$ is the conditional expectation map. As a consequence,  $\mathcal{T}_I$ is a densely defined, faithful, and lower semicontinuous trace on $\mathcal{I}_I$. Thanks to all these properties, one can define for a tuple ${\bf w}=(w_1,\dots,w_m)$ of  vectors in $\R^D$ the \emph{smooth} Chern cocycle

\begin{equation}
       {\rm Ch}^s_{I,{\bf w}}(f_0,f_{1},\dots, f_m)\;=\;\sum_{\rho\in S_m}(-1)^\rho \mathcal{T}_I\big(f_0\nabla _{w_{\rho(1)}}f_{1}\cdots \nabla _{w_{\rho(m)}}f_m\big)
    \end{equation}
for $(m+1)$-tuple of elements of suitable element $f_j\in\mathcal{I}_I$. Here, the directional derivative is given in terms of the formal Fourier series
$$\nabla_v f\;=\;\sum_{n\in \Z^D} v\cdot \nabla(f_nu^n)\;=\;- {\rm i}\sum_{i=1}^D\sum_{n\in \Z^D} v_in_if_n u^n$$
\begin{proposition}\label{prop: num}
 For any linearly independent tuple $(w_1,...,w_m)$ of vectors orthogonal to $v_I$ there is an elements$ [\zeta]_i\in K_i(\mathcal{I}_I)$, $i= m\, \mathrm { mod }\,2$, such that the pairing $\langle [\zeta]_i,[{\rm Ch}_{I,{\bf w}}]\rangle$ does not vanish.
\end{proposition}
\begin{proof}
From \cite[Proposition 4.1]{DET2} one gets $\mathcal{I}_I\simeq C(\T^D)\rtimes \R^{I}.$ Thus,  the Connes-Thom isomorphism  and \cite[Theorem 5.3]{DET2} completes the proof.
\end{proof}

\medskip
\noindent
{\bf Proof of Proposition \ref{prop: numerical invariants}.}
\begin{proof}
The idea of the proof is to construct a homomorphism $\iota_*\colon  K_*(\mathcal{I}_I)\to K_*(\mathfrak{I}_I)$ such that 
\begin{equation}\label{eq: equality chern}
     \big \langle [\zeta]_i, [{\rm Ch}^s_{I,{\bf w}}
      ]\big\rangle \;=\;\big\langle \iota_*[\zeta]_i, [{\rm Ch}_{I,{\bf w}}]\big\rangle, \qquad [\zeta]_i\in K_i(\mathcal{I}_I)
  \end{equation}  
  Hence, by combining this equality with Proposition \ref{prop: num}, the non-triviality is established.
    Let us start with the case for which ${\bf v}$ is CI.  Observe that the map $\Gamma $ provided in Lemma \ref{coro: surjective} induces a continuous surjective $\Z^D$-equivariant map $\tilde{\Gamma}_I\colon \tilde{\mathcal{C}}_{{\bf v}_I}\to \R^{I}$. By functoriality of the crossed product, there is an injective $*$-homomorphism   $\iota'\colon \mathcal{I}_I\to \tilde{\mathfrak{I}}_I$. From Remarks \ref{rem: basis of measure} and \ref{rem: basis of measure zd}, the Lebesgue measure $\nu_I$ on $\R^{I}$ is the pushforward measure of $\tilde{\mu}_I$. As a consequence,  $\mathcal{T}_I(f)\;=\;\tilde{\Tt}_I\big(\iota'(f)\big)$ for any $f\in \mathcal{I}_I.$ This implies the equality in Chern numbers
\begin{equation*}
     \big \langle [\zeta]_i, [{\rm Ch}^s_{I,{\bf w}}
      ]\big\rangle \;=\;\big\langle \iota'_*[\zeta]_i, [\tilde{\rm Ch}_{I,{\bf w}}]\big\rangle, \qquad [\zeta]_i\in K_i(\mathcal{I}_I)
  \end{equation*}
  Letting $\iota_*$ as the composition of the arrows $K_*(\mathcal{I}_I)\stackrel{\iota'_*}\to K_*(\tilde{\mathfrak{I}}_I)\simeq K_*(\mathfrak{I}_I)$ one lands in \eqref{eq: equality chern}, where the last isomorphism follows from the fact that  $\mathfrak{I}_I$ and $\tilde{\mathfrak{I}}_I$ are Morita equivalent by Theorem \ref{Teo morita}.\\
If  $\X_{{\bf v}_I}$ is discrete then $\tilde{\Gamma}_I$ provides an equivariant homeomorphism of $\tilde{\mathcal{C}}_{{\bf v}_I}$ onto the $\Z^D$-orbit in $\R^{I}$ of an single element $\omega\in \R^{I}$. With the corresponding surjective $*$-homomorphism $\iota_\omega\colon \mathcal{I}_I\to \tilde{\mathfrak{I}}_I$ one has by \cite[Proposition 5.5]{DET2}
 \begin{equation*}
     \big \langle [\zeta]_i, [{\rm Ch}^s_{I,{\bf w}}
      ]\big\rangle \;=\;\big\langle (\iota_\omega)_*[\zeta]_i, [\tilde{\rm Ch}_{I,{\bf w}}]\big\rangle, \qquad \forall [\zeta]_i\in K_i(\mathcal{I}_I)
  \end{equation*}
Thus, as in the previous case, $\iota_*$ is the composition of $(\iota_\omega)_*$ with $ K_*(\tilde{\mathfrak{I}}_I)\simeq K_*(\mathfrak{I}_I)$.
\end{proof}

\medskip 
Let us finish this section with a discussion about the consequences of Theorem \ref{Theorem 2} in the \emph{bulk-edge} correspondence for two-dimensional lattice models in the irrational case. Consider in $\R^2$ the half-plane geometry defined by the semigroup $\LL_{v}$ for a single vector $v\in \R^2$ with components rational independent over $\Q$. According to Proposition \ref{prop: containing 2}, the transversal hull is
\begin{align*}
\Xi_{v}\;&=\; \big\{ \LL^\emptyset_{v,x}\big\}_{x\in \R_+}\cup  \big\{ \LL^{\{1\}}_{v,x}\big\}_{x\in \mathtt{X}^{\{1\}}_{{\bf v}}} \cup \{ \Z^2\}\\
\;&=\; \big\{ \LL^\emptyset_{v,x}\big\}_{x\in \R_+}\cup \big\{\LL_{v,0}^{\{1\}}-n\;:\;n\in \LL_{v},\;n\neq 0\big\}\cup \{ \Z^2\}\\
\;&=\; \big\{ \LL^\emptyset_{v,x}\big\}_{x\in \R_+}\cup \mathscr{O}(\LL_{v,0}^{\{1\}})\cup \{ \Z^2\},
\end{align*}
since in this case $\mathtt{X}^{\{1\}}_{\bf v} = v\cdot (\LL_v \setminus \{0\})$. The canonical filtration $\{\Z^2\}=\Xi_0\subset \Xi_1=\Xi_v$ has length two.  The cofiltration of $\A_v$ reduces to the short  exact sequence
\begin{equation}\label{eq: seq bulk-boundary}
    0\to\mathfrak{I}_1\to \A_v\to \A_0\to 0
\end{equation}
where  $\mathfrak{I}_1= C^*_r(\LL_{v}|_{\mathcal{C}_v})$ and $\A_0\simeq C(\T^2)$ is the bulk algebra. This sequence is the core of the bulk-edge correspondence  \cite{PRO}.

\medskip

 Recall that  $\mathcal{I}_1=C_0(\R)\rtimes_{\beta} \Z^2\simeq C(\T^2)\rtimes \R$ is the smooth boundary algebra. Here the homomorphism $\iota_*\colon K_*(\mathcal{I}_1)\to K_*(\mathfrak{I}_1)$ provided in the proof Proposition  \ref{prop: numerical invariants} is surjective. The latter is a consequence of \cite[Corollary 2.8]{RX}, where it is shown that the map $\iota'_*\colon K_*(\mathcal{I}_1)\to K_*(\tilde{\mathfrak{I}}_1)$ is surjective, where $\tilde{\mathfrak{I}}_1$ is the extended boundary algebra. Thus, $\iota_*$ is also surjective. In particular, one has the relations
$$K_0(\mathfrak{I}_1)\;\simeq\;K_0(\mathcal{I}_1)\;=\;\Z^2,\qquad  K_1(\mathfrak{I}_1)\;\simeq\; K_0(\mathcal{I}_1)/\Z\;=\;\Z.$$ 
The isomorphism $K_1(\mathfrak{I}_1)\simeq \Z$ is given by a suitable normalization of the pairing with the Chern cocycle ${\rm Ch}_{\{1\},\{w\}}$ for $w$ orthogonal to $v$. Moreover, this pairing is completely determined by an input coming from the Chern cocycle on the bulk algebra $\A_0$:
\begin{proposition}\label{prop: duality}
  Let $\partial_1\colon K_i(\A_0)\to K_{i-1}(\mathfrak{I}_1)$ be  the connecting map  related to the sequence \eqref{eq: seq bulk-boundary}. One has the duality
   $$\langle [\zeta]_i,[{\rm Ch}_{\emptyset, \{w, v\}}]\rangle\;=\;\langle \partial_1([\zeta]_i),[{\rm Ch}_{\{1\},\{w\}}]\rangle \qquad \forall[\zeta]_i\in K_i(\A_0)$$
   where $w$ is any vector in $\R^2$ linearly independent to $v$.
\end{proposition}
\begin{proof}
Since $\A_0\simeq C(\T^2)$ and $\mathcal{J}_1\simeq C(\T^2)\rtimes\R$, one gets by \cite[Theorem 5.3]{DET2} the duality
 $$\langle [\zeta]_i,[{\rm Ch}_{\emptyset, \{w, v\}}]\rangle\;=\;\langle \partial_1^s([\zeta]_i),[{\rm Ch}^s_{\{1\},\{w\}}]\rangle \qquad [\zeta]_i\in K_i(\A_0)$$
 where $\partial_1^s\colon K_i(\A_0)\to K_{i-1}(\mathcal{I}_1)$  is the connecting map of \eqref{eq: cofiltration smooth}. Moreover, the above with \eqref{eq: equality chern} yield
  $$\langle [\zeta]_i,[{\rm Ch}_{\emptyset, \{w, v\}}]\rangle\;=\;\langle \iota_*\circ\partial_1^s([\zeta]_i),[{\rm Ch}_{\{1\},\{w\}}]\rangle $$
  To conclude the proof is enough to recall the definition of $\iota_*$ and to note that the following diagram is commutative
  $$ \xymatrix{0\ar[r]^{} & \mathcal{I}_{1} \ar[r]^{i}\ar[d]^{\iota'}& \mathcal{A}_1\ar[r]^{\phi_1}\ar[d]^{\iota'} &C(\T^2)\ar[r]^{}\ar[d]^{\iota'}&0\\
0\ar[r]^{} & \tilde{\mathfrak{I}}_1 \ar[r]^{i}& \tilde{\A}_1\ar[r]^{\tilde{\psi}_1} &C(\T^2)\ar[r]^{}&0 }$$
which implies $\partial_1=(\iota')_* \circ \partial^s_1=\iota_* \circ \partial^s_1$ due to naturalness of the connecting maps.
\end{proof}

We can now discuss the consequences of Proposition \ref{prop: duality} for lattice Hamiltonians. Let  $H=H^*\in M_N(\A_1)$ be a Hamiltonian with a bulk spectral gap at $0$, \ie  $0$ lies in a compact interval $\Delta$ contained in a spectral gap of the bulk Hamiltonian $H_B:=\psi_1(H)\in M_N(\A_0)\simeq M_N(C(\T^2)).$ One associates to $H$ the bulk invariant
$$[P_f]\;:=\;[\chi_{(-\infty,0)}(H_B)]_0\in K_0(\A_0)$$
defined by the \emph{Fermi projection}, and the edge invariant by the \emph{edge unitary}
$$[U_f]\;=\;[e^{i\pi f(H)}]_1 \in K_1\big(\mathfrak{I}_1\big)$$
where $f$ is any smooth function that takes the constant values $-1$ below $\Delta$ and $1$ above $\Delta$. It is not difficult to see that these invariants satisfy the correspondence $\partial_1([P_f])=[U_f]$. 
Therefore, Proposition \ref{prop: duality} leads to the bulk-edge correspondence 

$$\langle [P_f],[{\rm Ch}_{\emptyset, \{w, v\}}]\rangle\;=\;\langle [U_f],[{\rm Ch}_{\{1\},\{w\}}]\rangle$$
The numerical invariant on the right-hand side has a concrete physical interpretation in terms of quantized edge currents \cite{PRO}.

\section{On a real-space expression for the trace}\label{sec: trace}

Having seen that the trace $\Tt$ on each ideal $\mathfrak{I}=C_r^*(\LL_{\bf v}\rvert_{\mathcal{C}_{\bf v}})$ is essentially unique it is interesting to derive an explicit expression for it which can be computed from matrix elements in the canonical Hilbert space representation on $\ell^2(\LL_{\bf v})$. Such an expression is known in the case $d=1$ as the trace-per-surface-area, for reasons that will become clear. 

\begin{definition}
Denote the span of the linearly independent vectors ${\bf v} =\{v_1,\dots,v_d\}$ in $\R^D$ by $E_{\triangleleft}$ and its orthogonal complement by $E_{\triangledown}$. For the linear map $A_{\bf v}\colon \R^D \to \R^d$ consider the slab $$V_L\; =\; A_{\bf v}^{-1}([0,L]^d)\;\subset\; \R^D$$ and the as well as the window
$$W_t \;=\; \big\{n \in \R^D: \; \lVert n_{\triangledown}\rVert \leq t\big\}$$
where $n_{\triangledown}\in E_{\triangledown}$ is the orthogonal projection of $n$ and we used the Euclidean norm.  Define on $\ell^2(\LL_{\bf v})$ the projection operators $P_{L,t}$ on $\ell^2(\LL_{\bf v})$ restricting to all lattice points in the finite-volume slab
$$\Lambda_{L,t}\;=\; V_L \cap W_t.$$
For any $a\in C^*_r(\LL_{\bf v})$ that is trace-class w.r.t. $\Tt$ define
\begin{equation}\label{eq: limit1}
\begin{split}
   \hat{\Tt}(a)&\;=\;\lim_{L\to \infty} \lim_{t\to \infty} \frac{1}{t^{D-d} \mathrm{Vol}(B_{D-d})} \Tr(P_{L,t} a P_{L,t})
\\
&\;=\;\lim_{L\to \infty} \lim_{t\to \infty} \frac{1}{t^{D-d} \mathrm{Vol}(B_{D-d})} \sum_{n\in \Lambda_{L,t}\cap \LL_{\bf v}} \langle n| a|n \rangle. 
\end{split}
\end{equation}
with $\Tr$ the usual Hilbert-space trace on $\mathcal{B}(\ell^2(\LL_{\bf v}))$ and $\mathrm{Vol}(B_{D-d})$ the volume of the $(D-d)$-dimensional unit ball.
\end{definition}

\begin{remark}
As will become clear, the normalization is chosen such that the trace of the indicator function for the slab $V_L$ is approximately $$\hat{\Tt}(\chi_{V_L})\;\sim\; L^{D-d}\,\mathrm{Vol}(E_\triangleleft/ \langle v_1,\dots,v_d\rangle)$$
for large $L$ with the covolume of the rank $d$ lattice spanned by $v_1,\dots,v_d$. This corresponds exactly to the $d$-dimensional volume of $V_L$ projected to $E_{\triangleleft}$. For example, if $d=1$ then the ergodic average represents a trace per $(D-1)$-dimensional surface area (see Figure~\ref{fig:picture} for $D=2$). $\hfill \blacktriangleleft$
\end{remark}

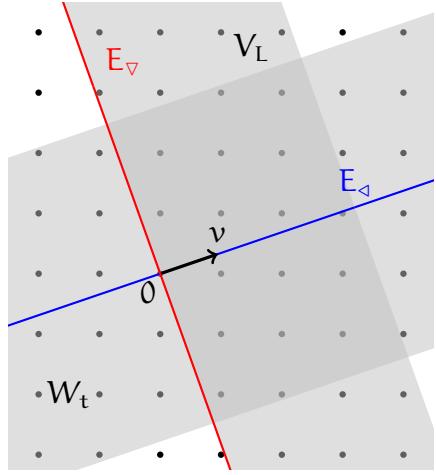
\begin{figure}
    \centering{}
\begin{tikzpicture}[scale=0.8]
    \clip (-2.5,-3.25) rectangle (4.5,4.5);
    \foreach \x in {-4,-3,-2,-1,0,1,2,3,4}
        \foreach \y in {-4,-3,-2,-1,0,1,2,3,4}
            \fill[black] (\x,\y) circle (1.5pt);
    
    \def\voneX{3.2}  
    \def\voneY{1.1}  
    \def\vtwoX{-1.1} 
    \def\vtwoY{3.1}  
    \def\vthreeX{0.945687}
    \def\vthreeY{0.32508}
    
    \def\svoneX{0.3 *\voneX}
    \def\svoneY{0.3 *\voneY}
    
    \def\svtwoX{1.5 *\vtwoX}
    \def\svtwoY{1.5 *\vtwoY}
    
    \coordinate (A) at ({0.8*\vtwoX-5*\svoneX}, {0.8*\vtwoY -5*\svoneY});
    \coordinate (B) at ({-0.8*\vtwoX-5*\svoneX}, {-0.8*\vtwoY-5*\svoneY});
    \coordinate (C) at ({-0.8*\vtwoX+6*\svoneX}, {-0.8*\vtwoY+6*\svoneY});
    \coordinate (D) at ({0.8*\vtwoX+6*\svoneX}, {0.8*\vtwoY+6*\svoneY});

    \fill[lightgray,opacity=0.5] (A) -- (B) -- (C) -- (D) -- cycle;

    \coordinate (E) at (-\svtwoX,-\svtwoY);
    \coordinate (F) at (\svtwoX,\svtwoY);
    \coordinate (G) at (-\svtwoX+ \voneX,-\svtwoY+ \voneY);
    \coordinate (H) at ({\svtwoX + \voneX}, {\svtwoY + \voneY});
    
    \fill[lightgray,opacity=0.5] (E) -- (G) -- (H) -- (F) -- cycle;
    
    \node (W) at (-1.5,-2) {$W_t$};
    \node (V) at (1.5,3.75) {$V_L$};
    \node (O) at (-0.2,-0.3) {$0$};
    \draw[thick,blue] (-\voneX,-\voneY) -- (\voneX,\voneY) node[anchor=south] {$E_{\triangleleft}$};
    \draw[thick,blue] (\voneX,\voneY) -- (2*\voneX,2*\voneY);
    \draw[very thick,->,black] (0,0) -- (\vthreeX,\vthreeY) node[anchor=south] {$v$};
    \draw[thick,red] (-2*\vtwoX,-2*\vtwoY) -- (\vtwoX,\vtwoY) node[anchor=south west] {$E_{\triangledown}$};
    \draw[thick,red] (\vtwoX,\vtwoY) -- (3*\vtwoX,3*\vtwoY);

\end{tikzpicture}

\caption{The geometry of the situation with $D=2$, $d=1$: $\LL_{\bf v}$  is a half-space with irrational normal vector $v$ and $V_L$ an infinite strip. The projection of $V_L$ onto $E_{\triangleleft}$ is dense in a line segment and the trace of the indicator function $\chi_{V_L}$ should be given by the length of that segment. We compute it by an ergodic average over increasing windows $W_t$.}
    \label{fig:picture}
\end{figure}

We will now prove that the limit \eqref{eq: limit1} exists for almost all choices of $\bf v$ and is proportional to the trace $\Tt$, thereby giving it a concrete interpretation. 

\medskip

Recall that there is a conditional expectation $E\colon \mathfrak{I} \to C_0(\mathcal{C}_{\bf v})$, where we consider $C_0(\mathcal{C}_{\bf v})$ as a commutative sub-algebra of $\mathfrak{I}$. The trace factors through this map $\Tt=\Tt \circ E$. In the groupoid picture the representation of any function $f\in C_c(\G_{\LL_{\bf v}})$ on $\ell^2(\LL_{\bf v})$ is given by the matrix elements
$$\langle n| \pi(f)|m \rangle\; =\; f(\LL_{{\bf v}}-n,m-n), \qquad n,m\in \LL_{\bf v},$$
in particular, all off-diagonal elements vanish in the case where $f\in C_0(\mathcal{C}_{\bf v})$ is an element of the commutative sub-algebra. In that case, the diagonal elements
$$\langle n| \pi(f)|n \rangle \;=\; f(\LL_{{\bf v}}-n), \qquad a\in C_0(\mathcal{C}_{\bf v}),n\in \LL_{\bf v}$$
sample a dense subset of the domain $\mathcal{C}_{\bf v}$.  In this representation, the conditional expectation $E$ acts by truncating the off-diagonal matrix elements
$$\langle n| \pi(E(f))|m \rangle\;=\;\langle n| \pi(f)|m \rangle\, \delta_{n,m}.$$

Thus we see that both $\Tt=\Tt\circ E$ and $\hat{\Tt}=\hat{\Tt}\circ E$ factor through the conditional expectation. It is therefore enough to prove that $\Tt$ and $\hat{\Tt}$ coincide on the commutative algebra $C_0(\mathcal{C}_{\bf v})\cap L^1(\mathcal{C}_{\bf v})$. Let us prepare the completely irrational case first. 
\begin{lemma}
\label{lemma:almost sure continuity}
Assume that $\X_{\bf v}= \R^d_+$. Define for $f\in C_0(\mathcal{C}_{\bf v})$ the function $\tilde{f}\colon \R^d_+\to \C$ via
$$\tilde{f}(x) := f(\LL_{{\bf v},x}^\emptyset).$$
Then $\tilde{f}$ is continuous outside a set of Lebesgue measure $0$.
\end{lemma}
\begin{proof}
We already saw that $N:=\bigcup_{J\neq \emptyset} \mathtt{X}^J_{\bf v}$ has measure $0$ and it is a consequence of Lemma~\ref{Lemma J} that $\lim_{y\to x}\LL_{{\bf v},y}^\emptyset=\LL_{{\bf v},x}^\emptyset$ if $x\notin N$, hence $\tilde{f}$ is continuous outside $N$.
\end{proof}

To relate the average with the integral we need a precise estimate for the number of lattice points inside a slab like $\Lambda_{L,t}$. We obtain it by specializing a recent result from \cite{KoivusaloLagace}:
\begin{theorem}
\label{th:lattice point count}
Let $\R^D=E_{\triangleleft}\oplus E_{\triangledown}$ be an orthogonal decomposition which is irrational w.r.t. $\Z^D$. Assume that $\Z^D$ is $\psi$-repellent w.r.t. that decomposition for a function $\psi:\R_+\to \R_+$ with $\psi(t)=O(t^\mu)$. For any set with finite perimeter $\Omega_{\triangleleft}\subset E_{\triangleleft}$ and $\delta>0$ there exists a constant $C$ such that
$$\lVert \#(\Z^D \cap (\Omega_{\triangleleft} \oplus E_{\triangledown}) \cap W_t) - \mathrm{Vol}(\Omega_{\triangleleft})\mathrm{Vol}(B_{D-d}) t^{D-d}\rVert\; \leq \;C t^{D-d - \frac{D-d}{D + \mu^{-1}}+\delta}$$
for all large enough $t$.
\end{theorem}

Accordingly, the number of lattice points in any slab windowed by $W_t$ behaves asymptotically in $t$ like its Euclidean volume. If one naturally parametrizes the possible decompositions by $(n\times d)$ matrices, then the assumption of $\psi$-repellence w.r.t. a suitable function are satisfied for a set whose complement has zero Lebesgue-measure \cite[Lemma 5.8]{BjorklundHartnick}, hence for almost all ${\bf v}$.

\begin{proposition}
\label{prop:riemann}
Assume that $\bf v$ is completely irrational $\X_{\bf v}=\R^d_+$ and that the technical assumption of Theorem~\ref{th:lattice point count} is satisfied.
For any $f\in C_0(\mathcal{C}_{\bf v})$ one has
$$\lim_{t\to \infty} \frac{1}{t^{D-d} \mathrm{Vol}(B_{D-d})} \sum_{n\in \Lambda_{L,t}\cap \LL_{\bf v}} \langle n| \pi(a)|n \rangle\;  =\; \frac{1}{\mathrm{Vol}(E_\triangleleft / \langle v_1,...,v_d\rangle)}\; \int_{[0,L]^d} \tilde{f}(x) d\nu(x)$$
with $\tilde{f}$ as in Lemma~\ref{lemma:almost sure continuity} and the normalized Lebesgue integral.
\end{proposition}

\begin{proof}
Recall that a function like $\tilde{f}$ which is continuous outside a set of measure $0$ is Riemann-integrable and its Riemann integral coincides with its Lebesgue integral. We will therefore, relate the sum to a Riemann integral for $\tilde{f}$. This requires us to subdivide $\Lambda_{L,t}$ into smaller slabs. For integer $M>0$ and $m\in [0,M)^d\cap \Z^d$ define 
$$\Lambda_{m,M,L,t}\;=\;A^{-1}_{\bf v}(R_{m,M,L})\cap W_t$$
with the boxes of sides $LM^{-1}$ given by $$R_{m,M,L}\;=\; \bigtimes_{i=1}^d[m_i,m_i+LM^{-1}]\;\subset \;\R^d.$$
We can bound for any $M$
\begin{align}
\label{eq:upper_riemann1}
&\frac{1}{t^{D-d} \mathrm{Vol}(B_{D-d})} \sum_{n\in \Lambda_{L,t}\cap \LL_{\bf v}} \langle n| \pi(a)|n \rangle \\ = &\frac{1}{t^{D-d} \mathrm{Vol}(B_{D-d})} \sum_{m\in [0,M)\cap \Z^d} \sum_{n\in \Lambda_{m,M,L,t}\cap \LL_{\bf v}} \langle n| \pi(a)|n \rangle \nonumber \\
\label{eq:upper_riemann2}
\leq &\sum_{m\in [0,M)\cap \Z^d}\frac{\#(\Lambda_{m,M,L,t}\cap \Z^D) }{t^{D-d} \mathrm{Vol}(B_{D-d})} \sup_{x\in R_{m,M,L}} \tilde{f}(x)
\end{align}
with $\#(\Lambda_{m,M,L,t} \Z^D)= \#(\Lambda_{m,M,L,t} \LL_{\bf v})$ counting the number of lattice points.

By Theorem~\ref{th:lattice point count} one has for fixed $M$ 
\begin{align*}
\lim_{t\to \infty} \frac{1}{t^{D-d} \mathrm{Vol}(B_{D-d})} \#(\Lambda_{m,M,L,t}\cap \Z^D) &= \lim_{t\to \infty} \frac{1}{t^{D-d} \mathrm{Vol}(B_{D-d})} \#(\Lambda_{m,M,L,t}\cap \Z^D)\\
&= \mathrm{Vol}_d(A^{-1}_{\bf v}(R_{m,M,L})\cap E_\triangleleft)\\
&= (L M^{-1})^{d}\mathrm{Vol}( A^{-1}_{\bf v}([0,1]^d)\cap E_\triangleleft)
\end{align*}
with the $d$-dimensional volume of $A^{-1}_{\bf v}(R_{m,M,L})\cap E_\triangleleft$ which can be computed as
$$\mathrm{Vol}_d(A^{-1}_{\bf v}([0,1]^d)\cap E_\triangleleft)=\lvert\det(A_{\bf v}\rvert_{E_\triangleleft})^{-1}\rvert=\frac{1}{\sqrt{\det(A_{\bf v}A_{\bf v}^t)}}= \frac{1}{\mathrm{Vol}(E_\triangleleft/\langle v_1,...,v_d\rangle)}$$
where we used that the Gram determinant computes the covolume of the lattice spanned by $v_1,...,v_d$. In conclusion, the right-hand side of \eqref{eq:upper_riemann2} converges for $t\to \infty$ to an upper Riemann sum for an equidistant partition of the square $[0,L]^d$.

Likewise we can bound \eqref{eq:upper_riemann1} from below by replacing the supremum in  \eqref{eq:upper_riemann2} by an infimum, hence
\begin{align*}
\sum_{m\in [0,M)\cap \Z^d} (LM^{-1})^{d} \inf_{x\in R_{m,M,L}}\tilde{f}(x) &\leq \lim_{t\to \infty} \frac{\mathrm{Vol}(E_\triangleleft/\langle v_1,...,v_d\rangle)}{t^{D-d} \mathrm{Vol}(B_{D-d})} \sum_{n\in \Lambda_{L,t}\cap \LL_{\bf v}} \langle n| \pi(a)|n \rangle\\
&\leq \sum_{m\in [0,M)\cap \Z^d} (LM^{-1})^{d} \sup_{x\in R_{m,M,L}} \tilde{f}(x). \end{align*}
For $M\to \infty$ both the upper and lower bound converge to the Riemann integral of $\tilde{f}$.
\end{proof}

\begin{corollary}
Assume the conditions of Proposition~\ref{prop:riemann} or that $\X_{\bf v}$ is discrete. Then there exists a constant $C$ such that
$$\Tt(a)\;=\; C \hat{\Tt}(a)$$
for all $a\in \mathfrak{I}$ which are $\Tt$-traceclass.
\end{corollary}
\begin{proof}
It only remains to consider the rational case since the other one is an obvious consequence of Proposition~\ref{prop:riemann}. In the rational case, $\Tt$ is induced by the counting measure on the discrete topological space $\mathcal{C}_{\bf v}$, \ie, concretely
$$\Tt(f)= {{\rm Vol}\big(\R^{d}/\tilde{\X}_{{\bf v}}\big)}\sum_{x\in \X_{\bf v}} f(\LL_{{\bf v},x}^\emptyset), \qquad \forall f\in C_c(\mathcal{C}_{\bf v})$$
with the normalization constant as in Remark~\ref{rem: basis of measure}.
Note that $V_{L}\cap \Z^D$ decomposes into a disjoint union over the fibers $$V_{x}:=A_{\bf v}^{-1}(\{x\})\cap \Z^D,$$ which are $(D-d)$-dimensional sublattices of $\Z^D$ (the lattices are translates of the kernel of the homomorphism $A_{\bf v}\rvert_{\Z^D}$ into a group isomorphic to $\Z^d$, hence it is a subgroup of $\Z^D$ of rank $D-d$).

Since any function $C_0(\mathcal{C}_{\bf v})$ decomposes into a sum of functions supported in a single $\LL_{{\bf v},x}^\emptyset$ we can reduce the computation of $\hat{\Tt}$ to the limit
$$\lim_{t\to \infty} \frac{1}{t^{D-d} \mathrm{Vol}(B_{D-d})} \sum_{n\in V_x\cap W_t} \langle n| \pi(f)|n \rangle=\lim_{t\to \infty} \frac{1}{t^{D-d} \mathrm{Vol}(B_{D-d})} \sum_{n\in V_x\cap W_t}  f(\LL_{{\bf v},x}^\emptyset)$$
where we used that $A_{\bf v}n = x$ for all $n$ in the sum.

Clearly,
$$\sum_{n\in V_x\cap W_t} f(\LL_{{\bf v},x}^\emptyset)= (\#V_x\cap W_t) f(\LL_{{\bf v},x}^\emptyset)$$
and we merely need to count those lattice points.  It is not difficult to show that the limit
$$\lim_{t\to \infty} \frac{1}{t^{D-d}}(\#V_x\cap W_t) =\frac{\mathrm{Vol}(B_{D-d})}{\mathrm{Vol}(E_{\triangledown}/V_x)}$$
exists and is equal to volume of the $(D-d)$-dimensional unit ball divided by the covolume of the lattice $V_x$. Note that this constant does not depend on $x$ since any two lattices $V_x$ and $V_{x'}$ are translates of each other and therefore have the same covolume. We conclude
$$
\hat{\Tt}(\pi(f))=\sum_{x\in \X_{\bf v}} \frac{1}{\mathrm{Vol}(E_{\triangledown}/V_x)} f(\LL_{{\bf v},x}^\emptyset) =  \frac{1}{\mathrm{Vol}\big(\R^{d}/\tilde{\X}_{{\bf v}}\big)\mathrm{Vol}(E_{\triangledown}/V_0)} \Tt(f).
$$
\end{proof}



\appendix

\section{}
In this appendix, we summarize the main consequences of the RCI property on ${\bf v}$ that we use in this work. Accordingly, throughout this section, we assume that ${\bf v}$ satisfies the RCI property.  We start with the following convergence criterion.

\begin{lemma}\label{Lemma J}
    Assume the RCI property on ${\bf v}$. Let $\overline{\R}:=\R\cup\{+\infty\}$ and $\{x(j)\}_{j\in \N}$ a sequence in $\X_{\bf v}$ with that converges to $x\in \overline{\R}^d.$ Denote by 
    \begin{enumerate}
    \item[(i)] $J_+\subset \{1,\dots,d\}$ the set of $k$ for which $\{x_k(n)\}_{n\in \N}$ is non-increasing and converges to a finite value,
    \item[(ii)] $J_-\subset \{1,\dots,d\}$ the set of $k$ such that $\{x_k(j)\}_{j\in \N}$ is strictly increasing and converges to a finite value,
    \item[(iii)] $J_\infty\subset \{1,\dots,d\}$ the set of $k$ such that $\{x_k(j)\}_{j\in \N}$ converges to $+\infty$.
    \end{enumerate} 
    If $J_+\cup J_-\cup J_\infty = \{1,\dots,d\}$ then $\LL_{{\bf v},x(j)}^\emptyset\to \LL_{{\bf v},x}^{J_+,J_-}$ in the Fell topology, where $ \LL_{{\bf v},x}^{J_+,J_-}$ is the semigroup
    consisting of all points $n\in\mathbb{Z}^D$ such that
\begin{align*}
v_k \cdot n + x_k &\;\geq\; 0, \qquad \forall k \in J_+,\\
v_k \cdot n + x_k &\;> \;0, \qquad \forall k \in J_-.
\end{align*}
\end{lemma}

\begin{proof}
Recall that $\LL_{{\bf v},x
(j)}^\emptyset$ is the subset of all $n\in \Z^D$ defined by the inequalities
$$v_k \cdot n + x_k(j) \geq 0, \qquad k=1,\dots,d.$$
 Since $J_+\cup J_-\cup J_\infty = \{1,\dots,d\}$ and the real numbers $x_k$ are approximated from below, respectively from above, the sets defined by the given inequalities clearly converge in the Fell topology.
\end{proof}

Now, let us assume that $D=d$ and ${\bf v}$ is not rational. This means that $\X_{{\bf v}_I}=\R^{d-|I|}_+$ for any properly contained subset $I$ of $\{1,\dots ,d\}$ while $\X_{\bf v}$ remains discrete. As a consequence, all the vectors ${\bf v}=\{ v_1,\dots,v_d\}$ have components linearly independents over $\Q.$ For any fixed $i\in \{1,\dots,d\}$, and $M,\epsilon>0$ define the set
$$\Lambda_{M,\epsilon,i}\;=\;\big\{n\in \Z^D: \epsilon> v_k\cdot n> 0, k\neq i \;\; \wedge \;\;  v_i\cdot n > M\big\}$$
An important property is that this set is non-empty:
\begin{proposition}\label{lemma_special_CI_case}
    Under the above assumption, the set $\Lambda_{M,\epsilon ,i}$ is non-empty for all $i\in \{1,\dots,d\}$ and $M,\epsilon >0.$
\end{proposition}
\begin{proof}
Due to $\X_{{\bf v}\setminus i}=\R^{d-1}_+$, each of the boxes $$B_{M,\epsilon}\;=\;\big\{n\in \Z^D\,|\,  \epsilon> v_k\cdot n> 0, k\neq i \;\; \wedge \;\; \lvert v_i\cdot n\rvert \leq M\big\},$$
and the slabs
$$S_{\epsilon}\;=\;\big\{n\in \Z^D:  \epsilon> v_k\cdot n> 0, k\neq i\big\}$$ contains infinitely many points. Choose some $m\in S_\epsilon$ and let $R>0$ be such that $m\in B_{R,\epsilon}$, e.g. $R=\lvert v_i\cdot n\rvert$. We set $$\tilde{\epsilon}\;=\;\frac{1}{2}\min_{k\neq i} (v_k\cdot m)\;>\;0.$$
Let $\tilde{R}>0$ be so large that $\tilde{R}-R > M$. There exists an element $\tilde{m} \in S_{\tilde{\epsilon}}\setminus B_{\tilde{R},\tilde{\epsilon}}$ since $S_{\tilde{\epsilon}}$ is infinite but $B_{\tilde{R},\tilde{\epsilon}}$ finite. 

There are now two possible cases: In the first case, $v_i\cdot \tilde{m}$ is positive, which means $v_i \cdot \tilde{m}>\tilde{R}>M$ and thus $\tilde{m} \in \Lambda_{M,\epsilon,i}$, showing that the set is non-empty.  In the second case, $v_i\cdot \tilde{m}< -\tilde{R}$ is negative. Consider instead $n=m-\tilde{m}$. One has $$0\;<\;\tilde{\epsilon} - \frac{1}{2}\tilde{\epsilon}\;<\;v_k\cdot n \;<\; \epsilon$$ for each $k\neq i$ as well as
$$v_i \cdot n \;= \;v_i \cdot m - v_i \cdot\tilde{m} \;>\; -R + \tilde{R} > M,$$
showing that $n\in \Lambda_{M,\epsilon,i}$. 
\end{proof}
\begin{lemma}\label{lemma convergence}
   Let the RCI property be valid for ${\bf v}.$  For any $I\subset \{1,\dots,d\}$ and $x\in \X_{{\bf v}_I}$ there exists a sequence $\{n(j)\}_{j\in \N}$ in $\LL_{\bf v}$ such that $A_{{\bf v}_I}n(j)\to x$, $j\mapsto v_k\cdot n(j)$ is non-increasing for $k\in I$ and $v_k\cdot n(j)\to +\infty$ for all $k\notin I.$
\end{lemma}
\begin{proof}
It is enough to prove existence of a sequence $\{n(j)\}_{j\in \N}$ in $\LL_{\bf v}$ such that $A_{{\bf v}_I}n(j)\to 0$ in $\R^{I}$, $j\mapsto v_k\cdot n(j)$ is non-increasing for $k\in I$ and $v_k\cdot n(j)\to +\infty$ for all $k\notin I.$ For general $x$ one can then add to $n(j)$ a suitable sequence $n'(j)$ such that $A_{{\bf v}_I}n'(j)\to x$ and $j\mapsto v_k\cdot n'(j)$ is non-increasing for $k\in I$ (and such a sequence $n'(j)$ always exists under RCI). It further suffices to prove the result for sets $I = \{1,\dots,d\}\setminus \{i\}$ for each $i \in \{1, \dots, d\}$, since the general case then follows by considering the sum of sequences associated to each $i\in \{1,...,d\}\setminus I$. 

If ${\bf v}$ is CI the result follows directly from the density of $A_{\bf v}(\LL_{\bf v})$ in $\mathcal{X}_{\bf v}$. If  $d=D$ and ${\bf v}$ is not rational, the existence follows by Proposition ~\ref{lemma_special_CI_case}.

The remaining case is the rational case, where we notice that ${\rm dim}_{\Z}({\rm Ker} A_{{\bf v} \setminus i}|_{\Z^D}) = \dim_{\Z}({\rm Ker}A_{{\bf v}}|_{\Z^D}) + 1$ since those are lattices of maximal rank. Therefore, there exists a non-zero vector $n \in \Z^D$ such that $A_{{\bf v} \setminus i} n = 0$ and $v_i\cdot n> 0$. Thus, the sequence $n(j) := jn$ for $j \in \N$ satisfies the required properties. 


\end{proof}

\end{document}